\documentclass[orivec,envcountsect,UKenglish,cleveref, autoref]{llncs}
\usepackage{hyperref}

\usepackage{comment}
\usepackage[utf8]{inputenc}

\usepackage{verbatim}
\usepackage{graphicx}
\usepackage{xcolor}
\usepackage{float}

\usepackage{amsthm}
\usepackage{amsfonts,amsmath,amssymb}
\usepackage{braket}
\usepackage{mathtools}
\usepackage{algorithmicx}
\usepackage[noend]{algpseudocode}
\usepackage[ruled,vlined,linesnumbered]{algorithm2e}

\usepackage{array}
\usepackage{multirow}
\usepackage{csquotes}
\usepackage{tabularx}

\usepackage{circuitikz}
\usepackage{tikz}
\usepackage{qcircuit}
\usepackage{mathtools}
\usepackage{pgfplots}

\usepackage{wrapfig}
\usepackage[export]{adjustbox}

\usepackage{subcaption}
\usepackage{slashbox}
\usepackage[capitalise,noabbrev]{cleveref}

\usepackage{caption}

\usetikzlibrary{
    calc,
    trees,
    positioning,
    arrows,
    chains,
    shapes.geometric, 
    decorations.pathreplacing,
    decorations.pathmorphing,
    shapes, 
    matrix,
    shapes.symbols
    }

\newcommand{\bigO}{\mathcal{O}}

\newcommand{\F}{\mathbb{F}}

\renewcommand{\vec}[1]{\mathbf{#1}}

\renewcommand{\Pr}{\mathbb{P}}

\newlength{\strutdepth}%
\newcommand{\notes}[3]{
	\noindent{%
		\color{#1}{[#3]}\color{#1}}%
	\strut\vadjust{\kern-\strutdepth%
		\vtop to \strutdepth{%
			\baselineskip\strutdepth%
			\vss\llap{{\large\color{#1}\textbf{#2}\quad\color{black}}}\null%
		}%
	}%
}

\newcommand{\Simon}{\textsc{Simon}}

\newcommand{\cnot}{\mathbf{cnot}}

\newcommand{\swap}{\mathbf{swap}}

\newcommand{\LPN}{\text{LPN}}

\newcommand{\LSN}{\text{LSN}}

\newcommand{\CN}{\textsc{CN}}
\newcommand{\IBMQ}{\textsc{IBM-Q16}}

\newcommand{\Zero}{\{\vec 0\}}
\newcommand{\sorth}{\vec s^\perp}

\definecolor{orcidlogocol}{HTML}{A6CE39}

\title{Noisy Simon Period Finding}
\titlerunning{Noisy Simon Period Finding}
\author{
	Alexander May\thanks{Funded by DFG under Germany's Excellence Strategy - EXC 2092 CASA - 390781972.
	} \href{https://orcid.org/0000-0001-5965-5675}{\protect\includegraphics[height=\fontcharht\font`B]{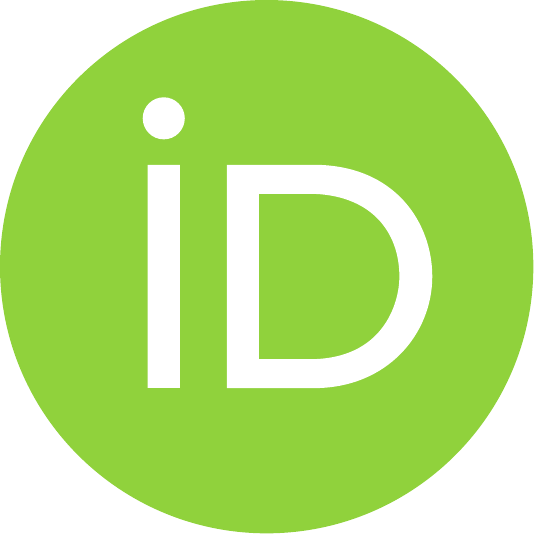}}\hspace*{1pt}
	\and
	Lars Schlieper$^\star$ \href{https://orcid.org/0000-0002-4870-1012}{\protect\includegraphics[height=\fontcharht\font`B]{orcid.pdf}}
	\and
	Jonathan Schwinger
	}
\institute{
	Horst G\"ortz Institute for IT Security \\
	Ruhr-University Bochum, Germany \\
	\email{\{alex.may,lars.schlieper,jonathan.schwinger\}@rub.de}}

\allowdisplaybreaks
\begin{document}

\maketitle

\begin{abstract}
Let $f: \F_2^n \rightarrow \F_2^n$ be a Boolean function with period $\vec s$. It is well-known that Simon's algorithm finds $\vec s$ in time polynomial in $n$ on quantum devices that are capable of performing error-correction. However, today's quantum devices are inherently noisy, too limited for error correction, and Simon's algorithm is not error-tolerant.

We show that even noisy quantum period finding computations may lead to speedups in comparison to purely classical computations.  To this end, we implemented Simon's quantum period finding circuit on the $15$-qubit quantum device IBM Q 16 Melbourne. Our experiments show that with a certain probability $\tau(n)$ we measure erroneous vectors that are not orthogonal to $\vec s$. We propose new, simple, but very effective smoothing techniques to classically mitigate physical noise effects such as e.g. IBM Q's bias towards the $0$-qubit.

After smoothing, our noisy quantum device provides us a statistical distribution that we can easily transform into an LPN instance with parameters $n$ and $\tau(n)$. Hence, in the noisy case we may not hope to find periods in time polynomial in $n$. However, we may still obtain a quantum advantage if the error $\tau(n)$ does not grow too large. This demonstrates that quantum devices may be useful for period finding, even before achieving the level of full error correction capability.

\keywords{Noise-tolerant Simon period finding, \IBMQ, LPN, quantum advantage}

\end{abstract}

%
%

\section{Introduction}

The discovery of Shor's quantum algorithm~\cite{FOCS:Shor94} for factoring and computing discrete logarithms in 1994 had a dramatic impact on public-key cryptography, initiating the fast growing field of post-quantum cryptography that studies problems supposed to be hard even on quantum computers,  such as e.g. Learning Parity with Noise (LPN)~\cite{FOCS:Alekhnovich03} and Learning with Errors (LWE)~\cite{STOC:Regev05}.

For some decades, the common belief was that the impact of quantum algorithms on symmetric crypto is way less dramatic, since the effect of Grover search can easily be handled by doubling the key size. However, starting with the initial work of Kuwakado, Morii~\cite{KuwakadoM12} and followed by Kaplan, Leurent, Leverrier and Naya-Plasencia~\cite{C:KLLN16} it was shown that (among others) the well-known Even-Mansour construction can be broken with quantum CPA-attacks~\cite{C:BonZha13} in polynomial time using Simon's quantum period finding algorithm~\cite{FOCS:Simon94}. This is especially interesting, because Even and Mansour~\cite{AC:EveMan91} proved that in the ideal cipher model any classical attack on their construction with $n$-bit keys requires $\Omega(2^{\frac n 2})$ steps. 

These results triggered a whole line of work that studies the impact of Simon's algorithm and its variants for symmetric key cryptography, including e.g.~\cite{SantoliSchaffner16,AC:LeaMay17,EC:AlaRus17,SAC:Bonnetain17,RSA:HosSas18,AC:BonNay18,DBLP:asiacrypt19}. In a nutshell, Simon's quantum circuit produces for a periodic function $f:\F_2^n \rightarrow \F_2^n$ with period $\vec s \in \F_2^n$, i.e. $f(\vec x)=f(\vec z)$ iff $\vec z \in \{\vec x, \vec x+ \vec s\}$, via quantum measurements uniformly distributed vectors $\vec y$ that are orthogonal to $\vec s$. It is not hard to see that from a basis of $\vec y$'s that spans the subspace orthogonal to $\vec s$,  the period $\vec s$ can be computed via elementary linear algebra in time polynomial in $n$. Thus, Simon's algorithm finds the period with a linear number of quantum measurements (and  calls to $f$),  and some polynomial  time classical post-processing. On any purely classical computer however, finding the period of $f$ requires in general $\Omega(2^{\frac n 2})$ operations~\cite{MontanaroW16}. 
Let us stress again that we consider quantum CPA attacks via Simon, i.e. the attacker has access to a cipher that is implemented quantumly---a very powerful attack model.

\paragraph{Our contributions.} 
We implemented Simon's algorithm on IBM's freely available \textsc{Q16 Melbourne}~\cite{IBMQ16}, called \IBMQ\ in the following, that realizes $15$-qubit quantum circuits. Since Simon's quantum circuit requires for $n$-bit periodic functions $2n$ qubits, we were able to implement functions up to $n=7$ bits. Due to its limited size, \IBMQ\ is not capable of performing full error correction \cite{calderbank1997quantum} for $n>1$. 
However, we show that error correction is no necessary requirement for achieving quantum speedups.

\medskip

\noindent {\bf Implementation.} Our experiments show that with some (significant) error probability $\tau$, we measure on \IBMQ\ vectors $\vec y$ that are {\em not orthogonal} to~$\vec s$. The error probability $\tau$ depends on many factors, such as the number of $1$- and $2$-qubit
gates that we use to realize Simon's circuit, \IBMQ's topology that allows only limited $2$-qubit applications, and even the individual qubits that we use. 
We optimize our Simon implementation to achieve minimal error $\tau$. Since increasing $n$ requires an increasing amount of gates, we discover experimentally that $\tau(n)$ grows as a function of $n$. For the function $f$ that we implemented, we found $\tau$-values ranging between $\tau(2)=0.09$ and $\tau(7)=0.13$.
We would like to stress that our choice of $f$ is highly optimized to minimize IBM-Q16's error. Any realistic real-word cryptographic $f$ would at the moment result in outputs close to random noise, i.e. with $\tau(n)$ close to $\frac 12$.

For our simple $f$ despite the errors we still qualitatively  observe the desired quantum effect: Vectors $\vec y$ orthogonal to $\vec s$ appear with significant larger probabilities than vectors not orthogonal to $\vec s$. Similar experimental observations have been achieved in Tame et al.~\cite{PhysRevLett.113.200501}. \\[-0.1cm]

\noindent {\bf Smoothing techniques.} In the error free case, Simon's circuit produces vectors that are uniformly distributed. However, on \IBMQ\ this is not the case. First, \IBMQ's qubits have different noise level, hence different reliability. Second, we experimentally observe vectors with small Hamming weight more frequently, the measured qubits have a bias towards $0$.

To mitigite both effects we introduce simple, but effective {\em smoothing techniques}. First, the quality of qubits can be averaged by introducing permutations that preserve the overall error probability $\tau$. Second, the $0$-bias can be removed by suitable addition of vectors, both quantumly and classically. In combination, our smoothing methods are effective in the sense that they provide a distribution where vectors orthogonal to $\vec s$ appear uniformly distributed with probability $1-\tau$, and vectors not orthogonal to $\vec s$ appear uniformly distributed with probability~$\tau$. Note that our smoothing techniques do not reduce the overall error $\tau$, but smooth the error distribution.

We call the problem of recovering $\vec s \in \F_2^n$ from such a distribution {\em Learning Simon with Noise} (LSN) with parameters $n$ and $\tau$.
Notice that  intuitively it should be hard to distinguish orthogonal vectors from non-orthogonal ones. \\[-0.1cm]

\noindent {\bf Hardness.}
We show that solving LSN with parameters $n, \tau$ is tightly polynomial time equivalent to solving the famous {\em Learning Parity with Noise} (LPN) problem with the same parameters $n, \tau$.
The core of our reduction shows that LSN samples coming from smoothed quantum measurements of Simon's circuit can be turned into perfectly distributed LPN samples, and vice versa. Hence, smoothed quantum measurements of Simon's circuit realize a {\em physical LPN oracle}. 

From an error-tolerance perspective, our \LPN-to-\LSN\ reduction may at first sound quite negative, since it is commen belief that we cannot solve \LPN\ (and thus also not \LSN) in time polynomial in $(n,\tau)$ --- not even on a quantum computer.
\\[-0.1cm]

\noindent{\bf Error Handling.} On the positive side, we may use the converse LSN-to-LPN reduction to handle errors from noisy quantum devices like $\IBMQ$ via LPN-solving algorithms. Theoretically, the best algorithm for solving LPN with constant $\tau$ is the BKW-algorithm of Blum, Kalai and Wasserman~\cite{STOC:BluKalWas00} with time complexity $2^{\bigO\big(\frac{n}{\log(\frac n {\tau})}\big)}$. This already improves on the classical time $2^{\frac n 2}$ for period finding.

Practically, the current LPN records  with errors $\tau \in [0.09, 0.13]$---as observed in our \IBMQ\ experiments---  are solved with variants of the algorithms \textsc{Pooled Gauss} and \textsc{Well-Pooled Gauss} of Esser, K\"ubler, May~\cite{C:EssKubMay17}. We show that \textsc{Pooled Gauss} solves LSN for $\tau \leq 0.292$ faster than classical period finding algorithms. \textsc{Well-Pooled Gauss} even improves on any classical period finding algorithm for all errors $\tau < \frac 1 2$.

\textsc{Well-Pooled Gauss} is able to handle errors in time $2^{cn}$, where $c<\frac 1 2$ is constant for constant $\tau$. For the error-free case $\tau=0$, we obtain polynomial time as predicted by Simon's analysis. In the noisy case $0 <\tau < \frac 1 2$ we achieve exponential run time, yet still improve over purely classical computation. This indicates that we achieve {\em quantum advantage} for the Simon period finding problem on sufficiently large computers, even in the presence of errors: Our quantum oracle helps us in speeding up computation! But as opposed to the exponential speedup from the (unrealistic) error-free Simon setting $\tau =0$, we obtain in the practically relevant noisy Simon setting $0 <\tau < \frac 1 2$ only a {\em polynomial speedup} with a polynomial of degree $\frac{1}{2c} >1$.

 Assume that in a possibly far future one could build a quantum device with $486$ qubits performing Simon's circuit on a $243$-bit {\em realistic real-world cryptographic} periodic function with error $\tau(486) = \frac 1 8$. Then our smoothed techniques could translate the noisy quantum data into an \LPN-instance with $(n,\tau)=(243, \frac 1 8)$. Such an LPN instance was solved in \cite{C:EssKubMay17} on 64 threads in only $15$ days, whereas classically period finding would require $2^{121}$ steps.
 
 We would like to stress that our introduction of a simple error parameter $\tau$ is to indicate at which point in the future quantum devices may help to speed up Simon-based quantum cryptanalysis. We do not give any predictions how $\tau(n)$ behaves for future devices, nor for realistic cryptographic functions. This remains an open problem. 
 
 \medskip
 
 Our paper is organized as follows. In \cref{sec:simon} we recall Simon's original quantum circuit, and already introduce our \LSN\ Error Model. In \cref{sec:ibmq} we run \IBMQ\ experiments, and show  in \cref{sec:smooth} how to smooth the results of the quantum computations\footnote{\IBMQ\ data can be found in our supplementary material.} such that they fit our error model. In \cref{sec:reduction} we show the polynomial time equivalence of \LSN\ and \LPN. In \cref{sec:correct} we theoretically show that quantum measurements with error~$\tau$ in combination with LPN-solvers outperform classical period finding for any $\tau < \frac 1 2$. Eventually, in \cref{sec:cexperiments} we experimentally extract periods from noisy \IBMQ\ measurements.

%
%

\section{Simon's Algorithm in the Noisy Case}
\label{sec:simon}

\paragraph{Notation.}

All $\log$s in this paper are base $2.$ Let $\vec x \in \F_2^n$ denote a binary vector with coordinates $\vec x = (x_{n-1}, \ldots, x_0)$ and Hamming weight $h(\vec x) = \sum_{i=0}^{n-1} x_i$. Let $\vec 0 \in \F_2^n$ be the vector with all-zero coordinates. We denote by ${\cal U}$ the uniform distribution over $\F_2$, and by ${\cal U}_n$ the uniform distribution over $\F_2^n$. If a random variable $X$ is chosen from distribution ${\cal U}$, we write $X \sim {\cal U}$. We denote by $\textrm{Ber}_{\tau}$ the Bernoulli distribution for $\F_2$, i.e. a $0,1$-valued $X \sim \textrm{Ber}_{\tau}$ with $\Pr[X=1] = \tau$.

Two vectors $\vec x, \vec y$ are {\em orthogonal} if their inner product $\langle \vec x, \vec y \rangle:= \sum_{i=0}^{n-1} x_i y_i \bmod~2$ is~$0$, otherwise they are called {\em non-orthogonal}. Let $\vec s \in \F_2^n$. Then we denote the subspace of all vectors orthogonal to $\vec s$ as
\[
  \vec s^{\perp} = \left\{ \vec x \in \F_2^n \; | \; \langle \vec x, \vec s \rangle = 0    \right\}.
\]
Let $Y=\{\vec y_1, \ldots, \vec y_k\} \subseteq \F_2^n$. Then we define $Y^{\perp} = \{\vec x \mid \langle \vec x, \vec y_i\rangle = 0 \textrm{ for all } i \}$.

For a Boolean function $f: \F_2^n \rightarrow \F_2^n$ we denote its {\em universal (quantum) embedding} by
\[
U_f: \F_2^{2n} \rightarrow \F_2^{2n} \textrm{ with } (\vec x, \vec y) \mapsto (\vec x, f(\vec x)+\vec y).
\]
Notice that $U_f(U_f(\vec x,\vec y)) = (\vec x, \vec y)$.

Let $\ket{x} \in \mathbb{C}^2$ with $x \in \F_2$ be a qubit. We denote by $H$ the {\em Hadamard function}
\[
  x \mapsto \frac 1 {\sqrt 2} (\ket{0} + (-1)^x \ket{1}).
\]
We briefly write $H_n$ for the $n$-fold tensor product $H \otimes \ldots \otimes H$. Let $\ket{x}\ket{y}\in \mathbb{C}^4$ be a $2$-qubit system. The $\cnot$ (controlled \textbf{not}) function is the universal embedding of the identity function, i.e. $\ket{x}\ket{y} \mapsto \ket{x}\ket{x+y}$. We call the first qubit $\ket{x}$ {\em control bit}, since we perform a \textbf{not} on $\ket{y}$ iff $x=1$.

A {\em Simon function} is a periodic $(2:1)$-Boolean function defined as follows.

\begin{definition}[Simon function/problem]\label{def:simon}
Let $f:\F_2^n\to \F_2^n$. We call $f$ a {\em Simon function} if there exists some period $\vec s \in \F_2^n \setminus\Zero$ such that for all $\vec x \not= \vec y \in \F_2^n$ we have
\[
  f( \vec x) = f( \vec y) \Leftrightarrow \vec y = \vec x + \vec s.
\]
In {\em Simon's problem} we have to find $\vec s$ given oracle access to $f$.
\end{definition}

In order to solve Simon's problem classically, we have to find some collision $\vec x \not= \vec y$ satisfying $f(\vec x) = f(\vec y)$. It is well-known that this requires $\Omega( 2^{\frac n 2})$ function evaluations.

Simon's quantum algorithm~\cite{FOCS:Simon94}, called \textsc{Simon} (see \cref{alg:simon}), solves Simon's problem with only $\bigO(n)$ function evalu\-ations on a quantum circuit. It is known that on input $\ket{0^{n}} \otimes \ket{0^{n}}$ a measurement of the first $n$ qubits of the quantum circuit $Q^{\Simon}_{f}$ depicted in \cref{circuit:simon} yields some $\vec y \in \F_2^n$ that is orthogonal to~$\vec s$. Moreover, $\vec y \in \F_2^n$ is uniformly distributed in the subspace ${\vec s}^{\perp}$, i.e. we obtain 
\begin{wrapfigure}{r}{0.5\textwidth}
	\centering
	\vspace*{-2em}
	\includegraphics[width=.5\textwidth]{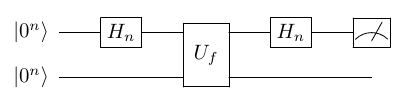}
	\caption{Quantum circuit $Q^{\Simon}_{f}$}\label{circuit:simon}
	\vspace*{-1em}
\end{wrapfigure}
each $\vec y \in {\vec s}^{\perp}$ with probability $\frac 1 {2^{n-1}}$.\linebreak
\textsc{Simon} repeats to measure  $Q^{\Simon}_{f}$ until it has collected $n-1$ linearly independent vectors $\vec y_1, \ldots, \vec y_{n-1}$, from which $\vec s$ can be computed via linear algebra in polynomial time. It is not hard to see that the collection of $n-1$ linearly independent vectors requires only $\bigO(n)$ function evaluations.

		\begin{algorithm}[htb]
			\DontPrintSemicolon
			\SetAlgoLined
			\SetKwInOut{Input}{Input}\SetKwInOut{Output}{Output}

			\SetKwComment{COMMENT}{$\triangleright$\ }{}

			\Input{
				Simon function $f: \F_2^n \rightarrow \F_2^n$.
			}
			\Output{
				Period $\vec s \in \F_2^n$
			}

			Set $Y=\emptyset$.\;
			\Repeat{$|Y|=n-1$}{
				Run $Q^{\Simon}_{f}$ on $\ket{0^n} \otimes \ket{0^n}$. \;
				Let $\vec y \in \F_2^n$ be the measurement of the first $n$ qubits. \label{line:choicey} \;
				If $\vec y \notin \textrm{span}(Y)$, then include $\vec y$ in $Y$.
			}
			Compute the unique $\vec s \in Y^{\perp} \setminus \{ \vec 0 \}$. \;
			\KwRet{$\vec s$.}

			\caption{\textsc{Simon}\label{alg:simon}}
		\end{algorithm}

At this point we should stress that \textsc{Simon} only works for {\em noiseless} quantum computations. Hence we have to ensure that each $\vec y$ is indeed in ${\vec s}^{\perp}$. Assume that we obtain in line~\ref{line:choicey} of algorithm \textsc{Simon} at least a single $\vec y$ with $\langle \vec y, \vec s \rangle=1$. Then the output of \textsc{Simon} is always false! Thus, \textsc{Simon} is not robust against noisy quantum computations.

More precisely, if we obtain in line~\ref{line:choicey} erroneous $\vec y \notin {\vec s}^{\perp}$ with probability $\tau$, $0 < \tau \leq \frac 1 2$, then \textsc{Simon} outputs the correct $\vec s$ only with exponentially small probability success probability $(1-\tau)^n$. This motivates our following quite simple error model.

\begin{definition}[LSN Error Model]
\label{def:error_model}
Let $\tau \in \mathbb{R}$ with $0 \leq \tau \leq \frac 1 2$. Upon measuring the first $n$ qubits of $Q^{\Simon}_{f}$, our quantum device outputs with probability $1-\tau$ some uniformly random $\vec y \in {\vec s}^{\perp}$, and with probability $\tau$ some uniformly random $\vec y \in \F_2^n \setminus {\vec s}^{\perp}$. That is, the output distribution is
\begin{equation}
\label{eq:LSN}
  \Pr[Q^{\Simon}_{f} \textrm{ outputs } \vec y] =
\begin{cases}
\frac{1-\tau}{2^{n-1}} & \textrm{if } \vec y \in {\vec s}^{\perp} \\
\frac{\tau}{2^{n-1}} & \textrm{else}
\end{cases}\;.
\end{equation}
We call $\tau$ the {\em error rate} of our quantum device. We call the problem of computing~$\vec s$ from the distribution in \cref{eq:LSN} {\em Learning Simon with Noise} (LSN). We further refine \LSN\  in \cref{def:LSN}.
\end{definition}

In the subsequent \cref{sec:ibmq} we show that the results of our \IBMQ\ implementation only roughly follows the LSN Error Model of \cref{def:error_model}. However, we also introduce in \cref{sec:smooth} simple smoothing techniques such that the \IBMQ\ measurements can be transformed into almost perfectly matching our error model.

Notice that intuitively there is no efficient way to tell whether $\vec y \in {\vec s}^{\perp}$. This intuition is confirmed in \cref{sec:reduction}, where we show that solving LSN is tightly as hard as solving the Learning Parity with Noise (LPN) problem.

%
%

\section{Quantum Period Finding on \IBMQ}
\label{sec:ibmq}
We ran our experiments on the \IBMQ\ Melbourne device, which (despite its name) realizes $15$-qubit circuits. 
Let us number \IBMQ's qubits as $0, \ldots, 14$. Our implementation goal was to realize quantum period finding for Simon functions $f: \F_2^n \rightarrow \F_2^n$ with error rate as small as possible. To this end we used the following optimization criteria.

\paragraph{Gate count.} \IBMQ\ realizes several $1$-qubit gates such as Hadamard and rotations, but only the $2$-qubit gate $\cnot$. On \IBMQ, the application of any gates introduces some error, where especially the $2$-qubit $\cnot$ introduces approximately as much error as ten $1$-qubit gates (see \cref{fig:IBM_topo}). Therefore, we introduce a circuit norm that defines a weighted gate count, which we minimize in the following.

\begin{definition}
\label{def:cnnorm}
	Let $Q$ be a quantum circuit with $g_1$ many 1-qubit gates and $g_2$ many 2-qubit gates.
	Then we define $Q$'s {\em circuit-norm} as
	$\CN(Q):=g_1 + 10 g_2.$
\end{definition}

\paragraph{Topology.} \IBMQ\ can only process 2-qubit gates on qubits that are adjacent in its topology graph, see \cref{fig:IBM_topo}. Let $G=(V,E)$ be the undirected topology graph, where node $i$ denotes qubit $i$.

\begin{figure}[htb]
	\begin{center}
		\includegraphics[width=\textwidth]{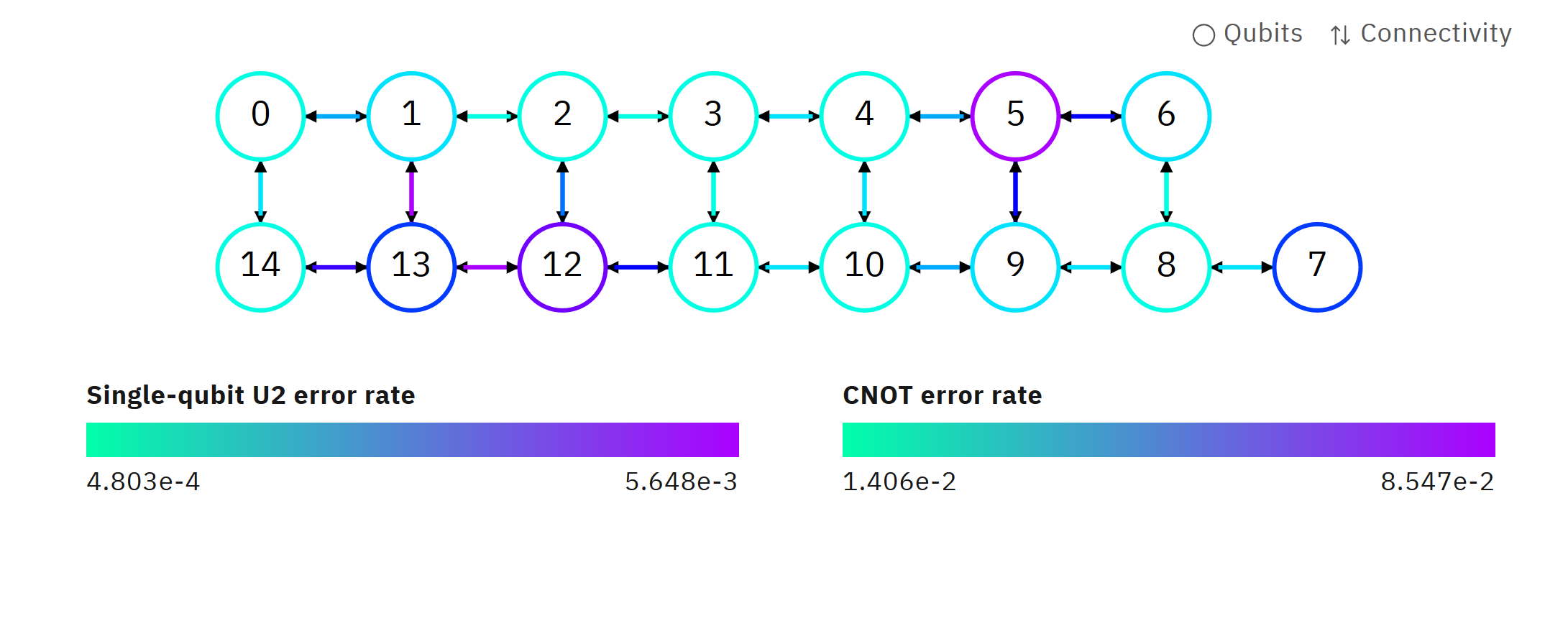}
	\end{center}
	\vspace*{-3em}
	\caption{Topology graph $G(V,E)$ of \IBMQ.}\label{fig:IBM_topo}
\end{figure}
\begin{wrapfigure}{r}{0.5\textwidth}
	\centering
	\vspace*{-2.5em}
	\includegraphics[width=.5\textwidth]{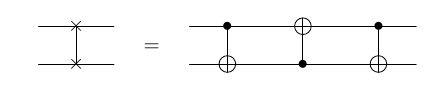}
	\caption{Realisation of \textbf{swap} via 3 $\cnot$s.}\label{fig:SWAP}
	\vspace*{-1em}
\end{wrapfigure}

If $\{u,v\} \in E$ then we can directly implement $\cnot(u, v)$, respectively $\cnot(v, u)$, where $u$, respectively $v$, serves as the control bit. Hence, we call qubits $u, v$ {\em adjacent} iff $\{u, v\} \in E$.

Let us assume that we want to realize $\cnot(1, 3)$ in our algorithm. Since $\{1,3\} \notin E$ we  may first swap the contents of qubits $2$ and $3$ by realizing a $\mathbf{swap}$ gate via 3 $\cnot$s as depicted in \cref{fig:SWAP}. Thus, with a total of 3 $\cnot$s we swap the content of qubit $3$ into $2$. Since $\{1,2\} \in E$, we may now apply $\cnot(1, 2)$.

\subsection{Function Choice}
\label{sec:fchoice}

	Notice that in \cref{def:simon} of Simon's problem, we obtain oracle access to a Simon function $f$.
	In a quantum-CPA attack we assume that a cryptographic function $f$ is realized via its quantum embedding $U_f$. An attacker gets black-box access to $U_f$, i.e. he can query $U_f$ on inputs of his choice in superposition.  
	
	We choose the following function $f_{\vec s}$ whose $U_{f_{\vec s}}$ is not too expensive to realize on \IBMQ .

\begin{definition}
	\label{def:fchoice}
Let $\vec s \in \F_2^n \setminus \{\vec 0\}$, and let $i \in [0,n-1]$ be the smallest $i$ with $s_i = 1$. We define
\[
  f_{\vec s}: \F_2^n \rightarrow \F_2^n, \quad \vec x \mapsto \vec x + x_i \cdot \vec s.
\]
\end{definition}

Let us first show that $f_{\vec s}$ is indeed a Simon function as given in \cref{def:simon}. Moreover, we show that {\em every}  Simon function -- no matter whether it is efficiently computable or not -- is of the form $f_{\vec s}$ followed by some permutation.

\begin{lemma}
\label{lem:f}
Let $f_{\vec s}(\vec x) = \vec x + x_i \cdot \vec s$ as in \cref{def:fchoice}. Then the following holds.
\begin{enumerate}
	\item[(1)] $f_{\vec s}$ is a Simon function with period $\vec s$, i.e. $f_{\vec s}(\vec x) = f_{\vec s}(\vec y)$ iff $\vec y \in \{\vec x , \vec x + \vec s\}$.
	\item[(2)] Any Simon function is of the form $P \circ f_{\vec s}$ for some bijection $P: \F_2^n \rightarrow \F_2^n$.
\end{enumerate}
\end{lemma}
\begin{proof}
	{(\em 1)} We have for all $\vec x \in \F_2^n$ that
	\[
	f_{\vec s}(\vec x + \vec s) = \vec x + \vec s + (\vec x + \vec s)_i \cdot \vec s = \vec x + \vec s + (x_i + 1) \cdot \vec s = \vec x + x_i \cdot \vec s = f_{\vec s}(\vec x).
	\]
	Thus, $f$ has period $\vec s$. It remains to show that $f_{\vec s}$ is $(2:1)$, i.e. that $f_{\vec s}(\vec x) = f_{\vec s}(\vec y)$ implies that $\vec y = \vec x$ or $\vec y = \vec x + \vec s$. From $f_{\vec s}(\vec x) = f_{\vec s}(\vec y)$ we conclude
	\[
	\vec x + x_i \cdot \vec s = \vec y + y_i \cdot \vec s.
	\]
	In the case $x_i=y_i$ this implies $\vec x = \vec y$, whereas in the case $x_i \not= y_i$ this implies $\vec y = \vec x + \vec s$. \medskip

	\noindent {(\em 2)} Let $g$ be an arbitrary Simon function with period $\vec s$. We have to write $g$ in the form $g = P \circ f_{\vec s}$. By ({\em 1}), we know that $f_{\vec s}(\vec x) = f_{\vec s}(\vec y)$ iff $\vec y \in \{ \vec x, \vec x + \vec s\}$. So $f_{\vec s}$ and $g$ already have the same arguments that collide. It remains to map $f_{\vec s}(\vec x)$ to the correct image $g(\vec x)$ via $P$. To this end define the bijection
\[
  P: \F_2^n \rightarrow \F_2^n , \ \vec x \mapsto
    g(\vec x + x_i \vec s).
\]
For all $\vec x \in \F_2^{n}$  we obtain
\begin{align*}
  P \circ f_{\vec s}(\vec x) = P(\vec x + x_i \vec s) & = g(\vec x + x_i \vec s + (\vec x + x_i\vec s)_i \cdot \vec s)  \\
  & = g(\vec x + x_i \vec s + (\vec x + 1)_i \cdot \vec s) = g(\vec x + \vec s) = g(\vec x),
\end{align*}
which implies $g = P \circ f_{\vec s}$.
\end{proof}

\paragraph{Instantiation of Function Choice.} Throughout the paper, we instantiate our function $f_{\vec s}$ with the period $\vec s = (s_{n-1}, \ldots, s_0) = 0^{n-2}11$ and $x_i=x_0$. We may realize $f_{\vec s}$ with $n$ $\cnot$-gates for copying $\vec x$, and an additional $2$ $\cnot$-gates for the controlled addition of $\vec s$ via control bit $0$. See \cref{fig:dida_circ} for an implementation of $f_{\vec s}$ with $n=3$.

\begin{wrapfigure}{r}{0.6\textwidth}
	\centering
	\includegraphics[width=.55\textwidth]{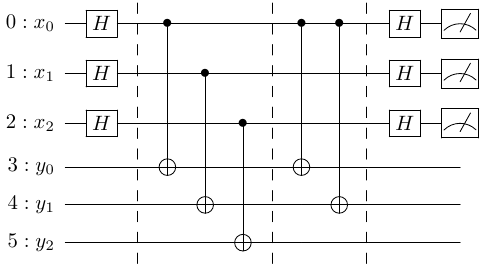}
	\caption{Simon circuit $Q_1$ with our realization of $f_{\vec s}$ and  $\CN(Q_1)=56$. The first 3 $\cnot$s copy $\vec x$, the remaining two $\cnot$s add $\vec s=110$.}\label{fig:dida_circ}
	\vspace*{-3em}
\end{wrapfigure}

Our function choice has the advantage that it can be implemented with only $n+2$ $\cnot$ gates (if we are able to avoid $\mathbf{swap}$s). In addition, we need $2n$ Hadamards for realizing \textsc{Simon}. Thus we obtain a small circuit norm $\CN=10(n+2) + 2n$, which in turn implies a relatively small error on \IBMQ. We perform further circuit norm minimization in \cref{sec:minimize}.

\paragraph{Discussion of our Simple Function Choice.} As shown in \cref{lem:f}, our function $f_{\vec s}$ is general in the sense that any Simon function is of the form $g = P \circ f_s(\vec x)$. However, for obtaining small circuit norm we instantiate our Simon function with the simplest choice, where $P$ is the identity function. In general, we could instantiate non-trivial $P$ via some variable-length PRF with fixed key such as SiMeck~\cite{CHES:YZSAG15}. This would however result in an explosion of the circuit norm and therefore in an explosion of \IBMQ's noise rate~$\tau(n)$.

\begin{wrapfigure}{l}{0.5\textwidth}
	\centering
		\vspace*{-2em}
	\includegraphics[width=.5\textwidth]{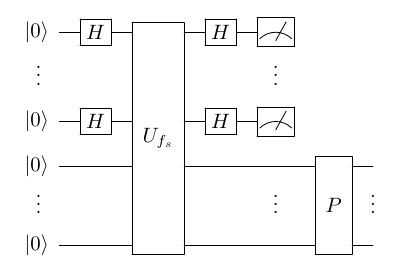}
	\caption{\textsc{Simon} with a general Simon function $P\circ f_{\vec s}$.}\label{fig:circ_with_permutation}
	\vspace*{-1.5em}
\end{wrapfigure}
Thus, \textsc{Simon} with a general Simon function could be implemented as depicted in~\cref{fig:circ_with_permutation}, where the permutation $P$ is quantumly implemented in-place on the last $n$ qubits (with at most one ancilla bit as shown in \cite{ShendePMH03}). 
But already from~\cref{fig:circ_with_permutation} one observes that $P$ does not at all effect the \textsc{Simon} algorithm. In fact, \textsc{Simon} outputs the measurement of the first $n$ qubits, which only depend on which arguments $\vec x, \vec x + \vec s$ collide under $f_{\vec s}$, but {\em not} which function value $f_{\vec s}(\vec x) = f_{\vec s}(\vec x + \vec s)$ they take (which is controlled by $P$).
So, quantumly the choice of a non-trivial $P$ would just unnecessarily increase the error rate $\tau$.

However, we would like to point out that choosing $P$ as the identity function implies that classically extract the period $\vec s$ is {\em not hard}. Notice that $f_{\vec s}(\vec x) \in \{\vec x, \vec x + \vec s\}$. Thus, we may compute $f_{\vec s}(1^{n}) + 1^{n} =\vec s$. The reason that $f_{\vec s}(\vec x)$ classically reveals its period so easily is that the image $\vec x + \vec s$ together with the argument $\vec x$ directly gives us $\vec s$. This correlation between argument $\vec x$ and image $\vec x + \vec s$ is destroyed by a random $P$, which explains why in general period finding classically becomes as hard as collision finding.

However, as explained above, \textsc{Simon} does {\em not profit} from a trivial $P$, since \textsc{Simon} is oblivious to concrete function values.

\subsection{Minimizing the gate count of $f_{\vec s}$}
\label{sec:minimize}

We may implement $f_{\vec s}$ on \IBMQ\ directly as the circuit $Q_1$ from \cref{fig:dida_circ}. Since $Q_1$ uses $6$ Hadamard- and $5$ $\cnot$-gates, we have circuit norm $\CN(Q_1) = 56$, but only when ignoring \IBMQ's topology. As already discussed, \IBMQ\ only allows $\cnot$s between adjacent qubits in the topology graph $G=(V,E)$ of \cref{fig:IBM_topo}.

Thus, \IBMQ\ compiles $Q_1$ to $Q_2$ as depicted in \cref{fig:swap_circ}. Let us check that $Q_2$ realizes the same circuit as $Q_1$, but only acts on adjacent qubits. Let $U_{f_{\vec s}}: \F_2^6 \rightarrow \F_2^6$ be the universal quantum embedding of $f_{\vec s}$ with $(\vec x, \vec y) \mapsto (\vec x, f_{\vec s}(\vec x) + \vec y) = \vec x + x_0 \vec s + \vec y)$. In $U_{f_{\vec s}}$ we first add each $x_i$ to $y_i$ via $\cnot$s, see \cref{fig:dida_circ}. Thus, we have to make sure that each $x_i$ is adjacent to its $y_i$.
Second, we add $\vec s = 011$ via $\cnot$s controlled by $x_0$. Thus, we have to ensure that $x_0$ is adjacent to $y_0$ and $y_1$.

We denote by $i: j$ that qubit $i$ contains the value $j$. This allows us to define the starting {\em configuration} as
\[
  0: x_0 \quad 1:x_1 \quad 2: x_2 \quad 3: y_0 \quad 4:y_1 \quad 5:y_2.
\]
Step 1 of $Q_2$ (see \cref{fig:dida_circ}) performs $\swap(2,3)$ and thus results in configuration
\[
  0:x_0 \quad 1:x_1 \quad 2:y_0 \quad 3:x_2 \quad 4:y_1 \quad 5:y_2.
\]
Step 2 of $C_2$ performs $\swap(1, 2)$ as well as $\swap(4, 3)$. This results in configuration
\[
  0:x_0 \quad 1:y_0 \quad 2:x_1 \quad 3:y_1 \quad 4:x_2 \quad 5:y_2.
\]
Since $\{0,1\},\{2,3\},\{4,5\} \in E$, in Step 3 we now compute $\cnot(0,1)$, $\cnot(2,3)$ and $\cnot(4,5)$.
This realizes the computation of $\vec x + \vec y$.
Eventually, Step 4 of $C_2$ performs $\swap(0, 1)$ and $\swap(2, 3)$ resulting in
\[
  0:y_0 \quad 1:x_0 \quad 2:y_1 \quad 3:x_1 \quad 4:x_2 \quad 5:y_2.
\]
For realizing the addition of $x_i \cdot \vec s = x_0 \cdot 011$, in Step 5 we compute $\cnot(1, 0)$ and $\cnot(1, 2)$ using $\{0,1\},\{1,2\} \in E$.

\begin{figure}[htb]
	\centering
	\vspace*{-1em}
	\includegraphics[width=\textwidth]{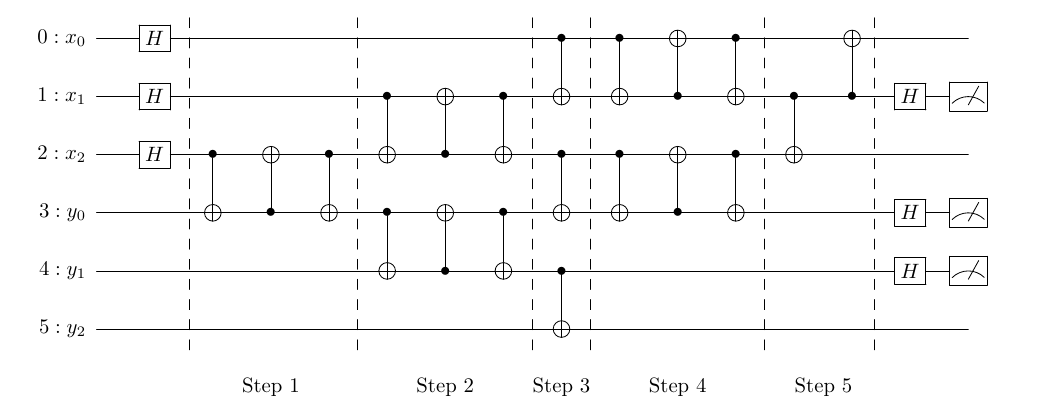}
	\caption{\IBMQ\  compiles $Q_1$ to $Q_2$ with $\CN(Q_2)=206$.}\label{fig:swap_circ}
	\vspace*{-1em}
\end{figure}
\begin{wrapfigure}{r}{0.55\textwidth}
	\centering
	\vspace*{-2em}
	\includegraphics[width=.55\textwidth]{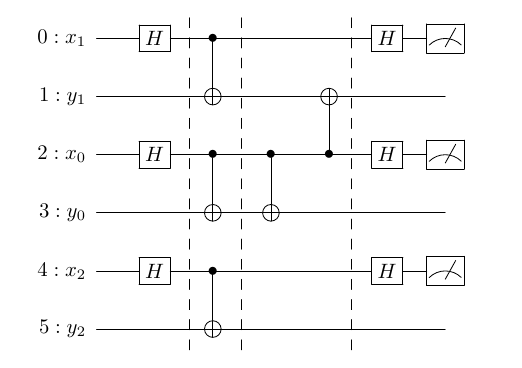}
	\caption{Circuit $Q_3$.}\label{fig:circ:opt_step}
	\vspace*{-1em}
\end{wrapfigure}
In total $Q_2$ consumes six $1$-bit gates and twenty $2$-bit gates and thus has $\CN(Q_2) = 206$, as compared to $\CN(Q_1)=56$.
In the following, our goal is the construction of a quantum circuit that implements $Q_1$'s functionality with minimal circuit norm on \IBMQ.
\begin{wrapfigure}{r}{0.55\textwidth}
	\centering
	\vspace*{1em}
	\includegraphics[width=.5\textwidth]{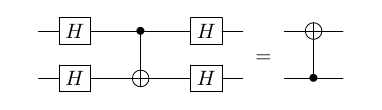}
	\caption{Control bit change.}\label{fig:CNOT}
	\vspace*{-1em}
\end{wrapfigure}

In \cref{fig:circ:opt_step} we start with circuit $Q_3$, for which our optimization eventually results in circuit $Q_4$ (\cref{fig:circ:opt_H}) that can be realized on \IBMQ\ with gate count only $\CN(Q_4)=33$.

From the discussion before, it should not be hard to see that $Q_3$ realizes $Q_{f_{\vec s}}^{\Simon}$, but yet it has to be optimized for \IBMQ.
First of all observe that $\cnot$ is self-inverse, and thus we can eliminate the two $\cnot(2,3)$ gates. Afterwards, we can safely remove qubit 3. The resulting situation for qubits $0, 1, 2$ is depicted in \cref{fig:circ:Optimierung}, where we use a control bit change (see \cref{fig:CNOT}).

\begin{figure}[htb]
	\vspace*{-0em}
	\centering
	\includegraphics[width=\textwidth]{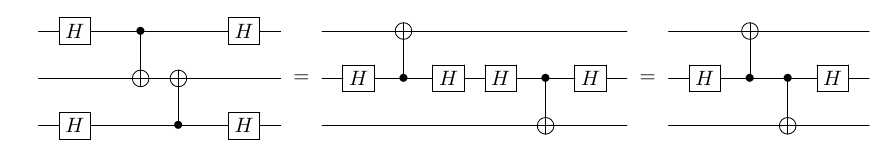}
	\caption{Optimization of $Q_3$.}\label{fig:circ:Optimierung}
	\vspace*{-1em}
\end{figure}

From \cref{fig:circ:Optimierung} we see that the change of control bits from $\cnot(0,1)$, $\cnot(2,1)$ to $\cnot(1,0)$, $\cnot(1,2)$ leads to some cancellation of self-inverse Hadamard 
gates. Moreover, the second
Hadamard of qubit $1$ can be eliminated, since it does not influence the measurement. We end up with circuit $Q_4$ with an optimized gate count of $\CN(Q_4) = 33$.

\begin{figure}[htb]
	\centering
	\includegraphics[]{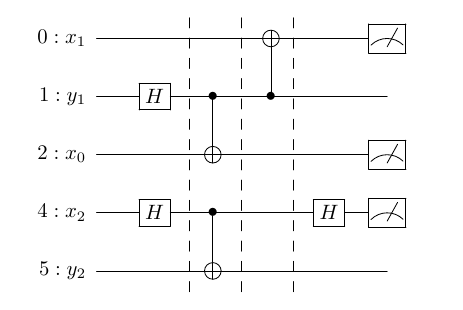}
	\caption{Optimized circuit $Q_4$ on \IBMQ\ with $\CN(Q_4)=33$.}\label{fig:circ:opt_H}
\end{figure}

Since $\{0,1\},\{1,2\},\{4,5\} \in E$, all three $\cnot$s of $Q_4$ can directly be realized on \IBMQ.
Notice that a configuration with optimal circuit norm is in general not unique. For our example, the following configuration yields the same circuit norm as the configuration of $Q_4$:\vspace*{-1em}

\[
	3:y_0 \quad 4:x_0 \quad 5:y_1 \quad 6:x_1 \quad 8:y_2 \quad 9:x_2.
\]

We optimized our \IBMQ\ implementation by choosing among all configurations with minimal circuit norm the one using \IBMQ's qubits of smallest error rate (see \cref{fig:IBM_topo}).
The choice of our configurations is given in Table~\ref{tab:config}, a complete list of optimized circuits of this table can be found in \cref{sec:appendix}, \cref{fig:collection_circuit}.

\newcolumntype{P}[1]{>{\centering\arraybackslash}p{#1}}
\begin{table}[htb]
	\centering
	\scalebox{1}{
	\begin{tabular}{|c|*{16}{P{0.045\textwidth}|}}\hline
		\backslashbox{$n$}{$q$}	&$0$	&$1$	&$2$	&$3$	&$4$	&$5$	&$6$	&$7$	&$8$	&$9$	&$10$	&$11$	&$12$	&$13$	&$14$	&$\mathbf{CN}$  \\\hline
		$2$						&$y_1$	&$x_0$	&$	$	&$	$	&$	$	&$	$	&$y_0$	&		&$	$	&$	$	&$	$	&$	$	&$	$	&$	$	&$x_1$	&$\mathbf{21}$	\\\hline
		$3$						&$y_1$	&$x_0$	&$	$	&$	$	&$	$	&$	$	&$y_0$	&		&$x_2$	&$y_2$	&$	$	&$	$	&$	$	&$	$	&$x_1$	&$\mathbf{33}$	\\\hline
		$4$						&$y_1$	&$x_0$	&$	$	&$	$	&$x_3$	&$y_3$	&$y_0$	&		&$x_2$	&$y_2$	&$	$	&$	$	&$	$	&$	$	&$x_1$	&$\mathbf{45}$	\\\hline
		$5$						&$y_1$	&$x_0$	&$	$	&$	$	&$x_3$	&$y_3$	&$y_0$	&		&$x_2$	&$y_2$	&$x_4$	&$y_4$	&$	$	&$	$	&$x_1$	&$\mathbf{57}$	\\\hline
		$6$						&$y_1$	&$x_0$	&$	$	&$	$	&$x_3$	&$y_3$	&$y_0$	&		&$x_2$	&$y_2$	&$x_4$	&$y_4$	&$x_5$	&$y_5$	&$x_1$	&$\mathbf{69}$	\\\hline
		$7$						&$y_1$	&$x_0$	&$x_6$	&$y_6$	&$x_3$	&$y_3$	&$y_0$	&		&$x_2$	&$y_2$	&$x_4$	&$y_4$	&$x_5$	&$y_5$	&$x_1$	&$\mathbf{81}$	\\\hline
	\end{tabular}
	}
	\vspace*{0.1cm}
	\caption{Table of configurations.}\label{tab:config}
\end{table}

\subsection{Experiments on IBM Q 16}
\label{sec:qexperiments}

For each dimension $n=2, \ldots , 7$ we took $8192$ measurements on \IBMQ\ of our optimized circuits from the previous section. The resulting relative frequencies are depicted in \cref{fig:simon_experiment}. For each $n$, let $S(n)$ denote the set of erroneous measurements in $\F_2^n \setminus \sorth$. Then we compute the error rate $\tau(n)$ as $\tau(n) = \frac{|S(n)|}{8192}$. In \cref{fig:simon_experiment} we draw horizontal lines $\frac{1-\tau(n)}{2^{n-1}}$, respectively $\frac{\tau(n)}{2^{n-1}}$, for the probability distributions of our LSN Error Model for orthogonal, respectively non-orthogonal, vectors.
\vspace*{-1em}
\begin{figure}[H]
	\begin{center}
		\begin{adjustbox}{minipage=\linewidth,scale=0.95}
		\begin{subfigure}[t]{0.49\textwidth}
			\hspace*{-0.2cm}
			\includegraphics[width=\textwidth]{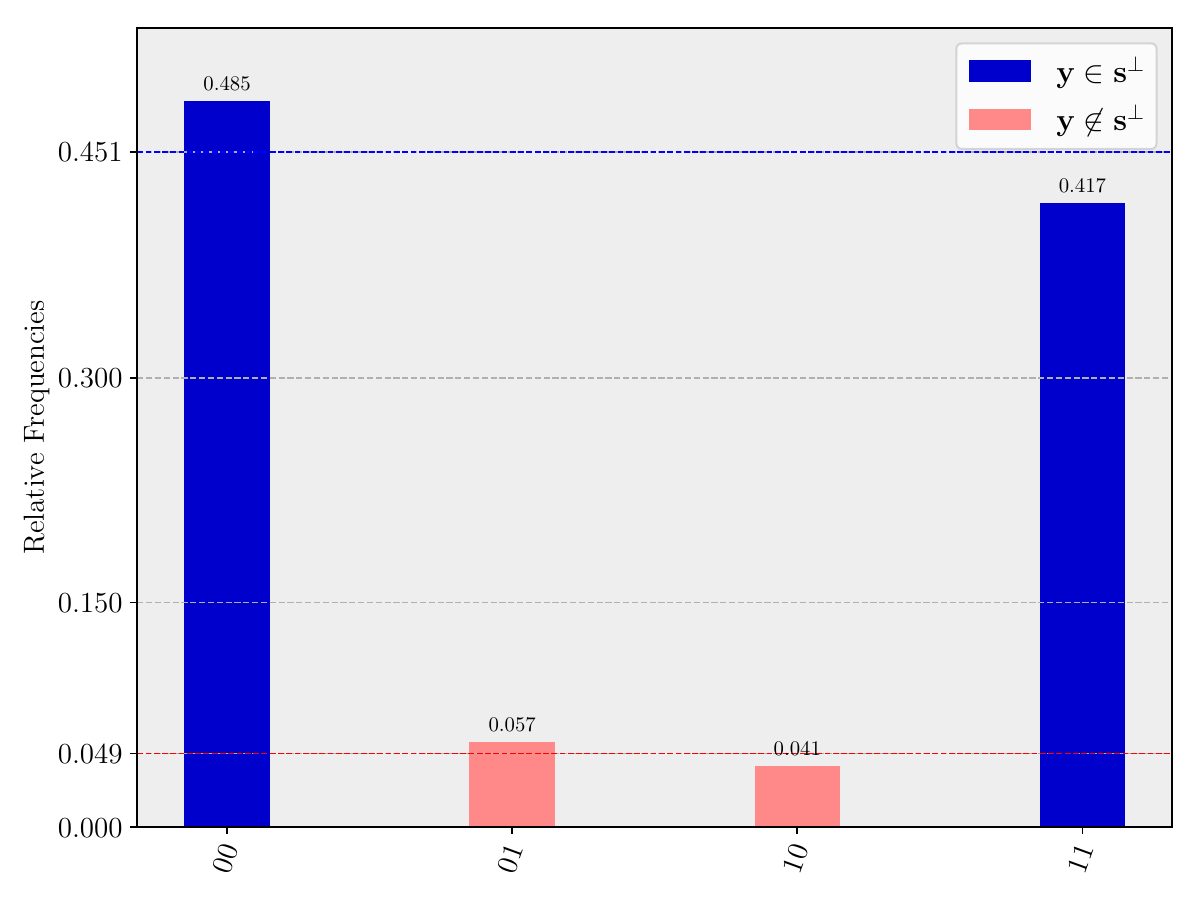}
			\caption{$\tau(2) = 0.099$}
		\end{subfigure}
		\hfill
		\begin{subfigure}[t]{0.49\textwidth}
			\hspace*{-0.2cm}
			\includegraphics[width=\textwidth]{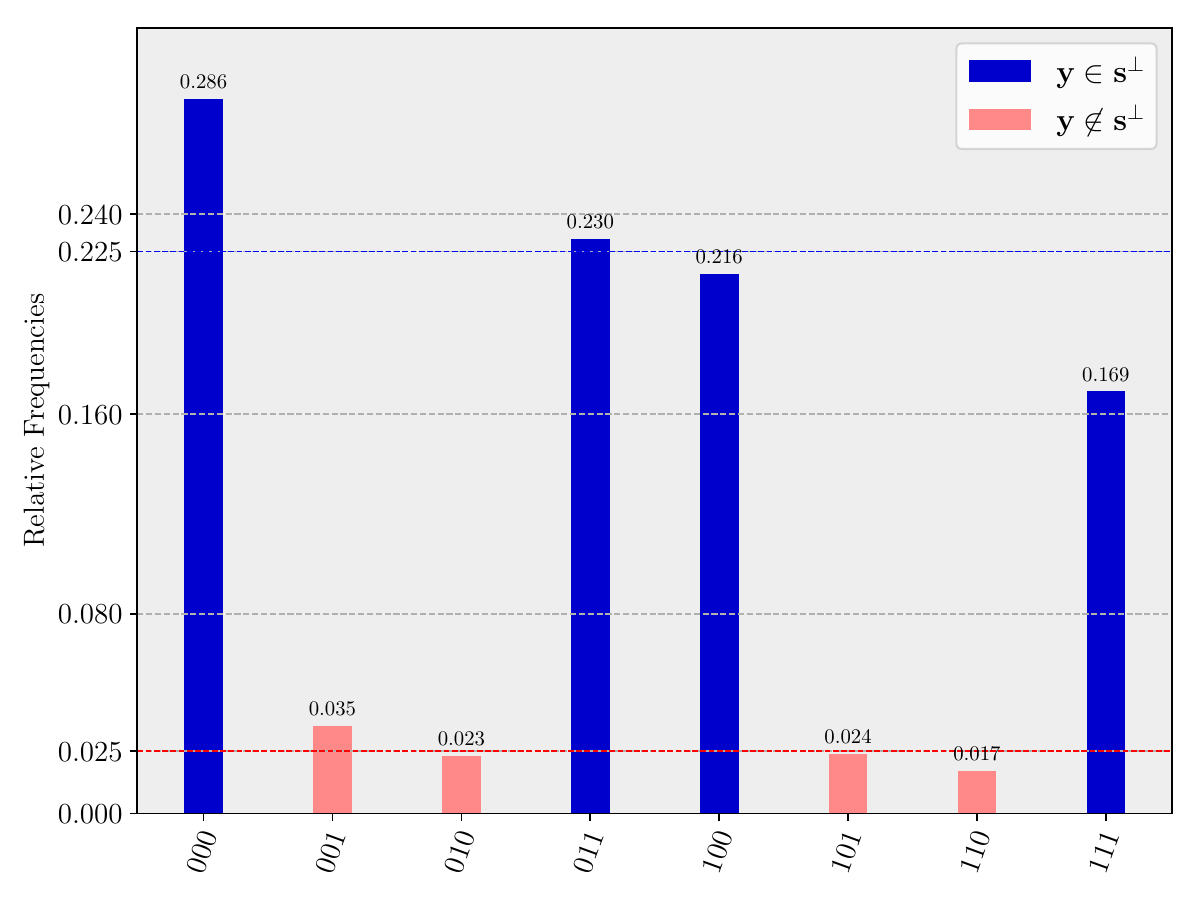}
			\caption{$\tau(3) = 0.098$}
		\end{subfigure}
		\\
		\begin{subfigure}[t]{0.49\textwidth}
			\hspace*{-0.2cm}
			\includegraphics[width=\textwidth]{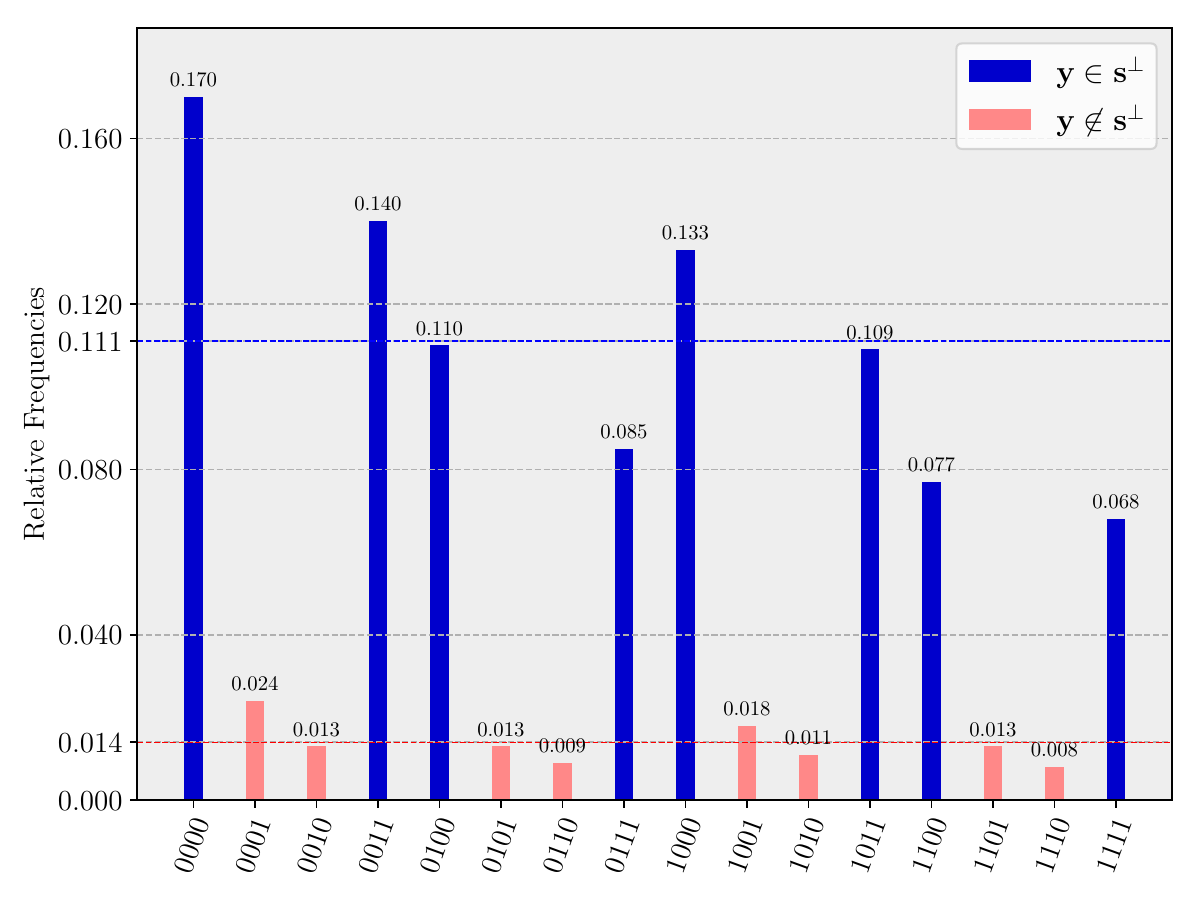}
			\caption{$\tau(4) = 0.102$}
		\end{subfigure}
		\hfill
		\begin{subfigure}[t]{0.49\textwidth}
			\hspace*{-0.2cm}
			\includegraphics[width=\textwidth]{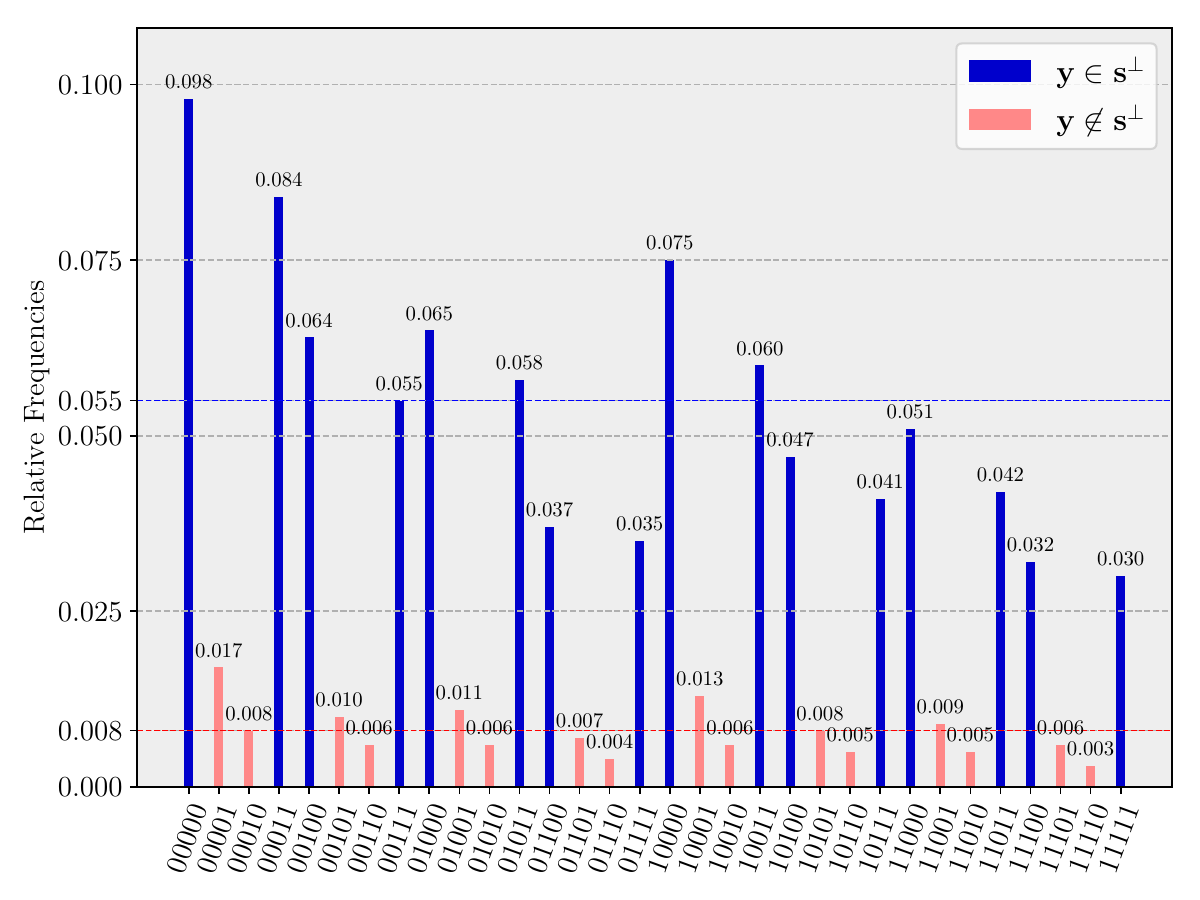}
			\caption{$\tau(5) = 0.107$}
		\end{subfigure}
		\\
		\begin{subfigure}[t]{0.49\textwidth}
			\hspace*{-0.2cm}
			\includegraphics[width=\textwidth]{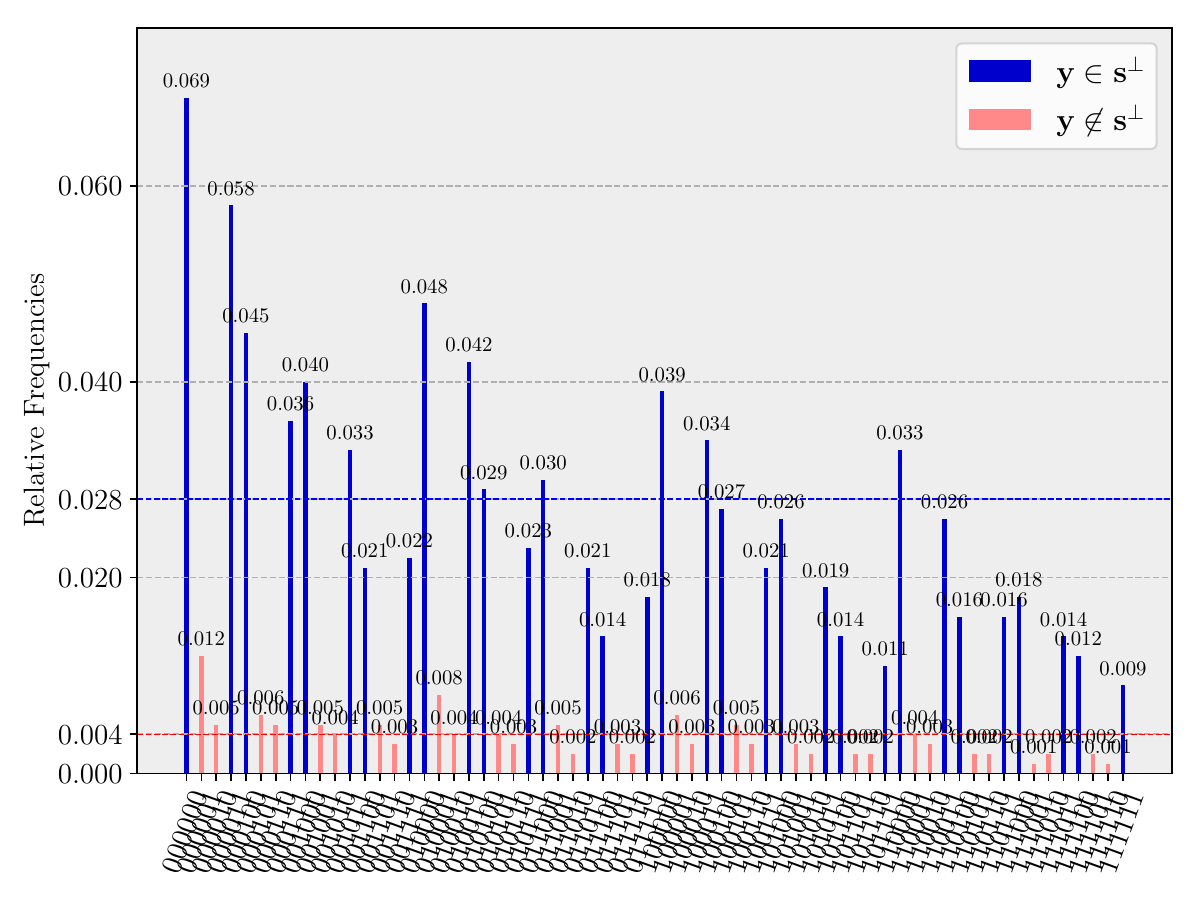}
			\caption{$\tau(6) = 0.112$}
		\end{subfigure}
		\hfill
		\begin{subfigure}[t]{0.49\textwidth}
			\hspace*{-0.2cm}
			\includegraphics[width=\textwidth]{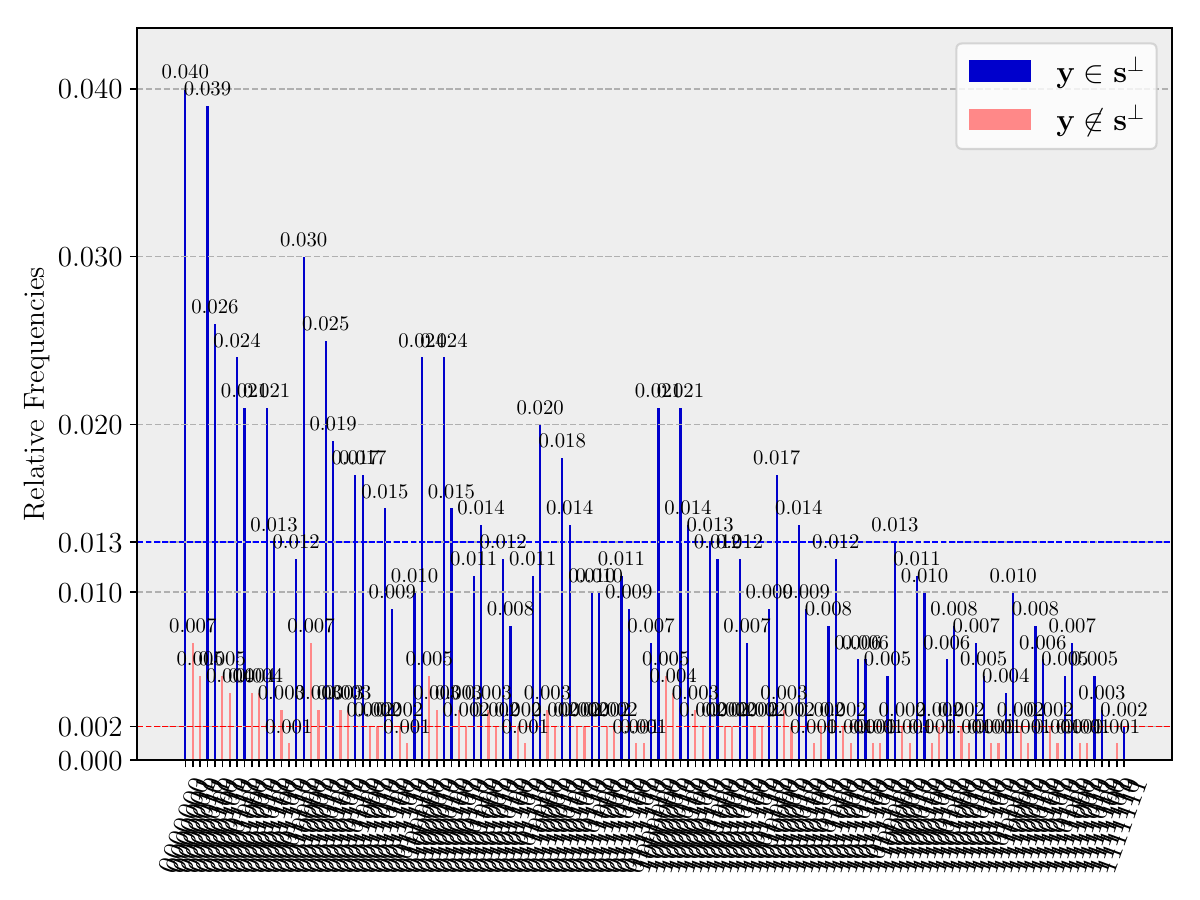}
			\caption{$\tau(7) = 0.117$}
		\end{subfigure}
		\end{adjustbox}
	\end{center}
	\vspace*{-1em}
	\caption{\IBMQ\  measurements of our optimized circuits (see \cref{sec:appendix}, \cref{fig:collection_circuit}).}\label{fig:simon_experiment}
	\vspace*{-1em}
\end{figure}

On the positive side, we observe that vectors in $\vec s^{\perp}$ are much more frequent. Hence,  $\IBMQ$ is noisy, but in principle works well  for period finding. E.g. for $n=3$, we have $\{\vec s\}^{\perp}= \{011\}^{\perp} = \{000, 011,100, 111\}$, and we measure one of these vectors with probability $1- \tau \approx 90 \%$.

On the negative side, we observe the following effects.
\begin{itemize}
\item {\bf Different qubit quality.}
We deliberately ordered our qubits by error rate to make the quality effect visible. Using the \IBMQ\  calibration, we choose lowest error rate for the least significant bit $x_0$ up to highest error rate for the most significant bit $x_{n-1}$ (nevertheless e.g. for $n=4$ it seems that the qubit for $x_2$ performed worse than the one for $x_3$).
\item{\bf Bias towards $\mathbf{0}$.} In \cref{fig:simon_experiment} we ordered our measurements on the $x$-axis lexicographically.
It can be observed that in general measurements with small Hamming weight appear with larger frequencies than large Hamming weight measurements. This indicates a 
bias towards the $\ket{0}$ qubit, which seems to be a natural physical effect since $\ket{0}$ is a non-activated ground state. 

\item {\bf Increasing $\tau(n)$.} The error rate $\tau(n)$ is a function increasing in $n$. This is what we expected, since the circuit norm increases with $n$, and for larger $n$ we also had to include lower quality qubits.
\end{itemize}

\begin{remark}
We experimented with different periodic $f_{\vec s}$, especially more complex than our choice from \cref{def:fchoice}. Qualitatively, we observed similar effects albeit with larger error rates $\tau(n)$. 
\end{remark}

The effects of {\em different qubit quality} and {\em bias towards $0$} obviously violate our LSN Error Model from~\cref{def:error_model}, since they destroy the uniform distribution among orthogonal, respectively non-orthogonal, vectors. 
However, we introduce in the subsequent \cref{sec:smooth} simple smoothing technique that (almost perfectly) mitigate both effects.

\section{Smoothing Techniques}\label{sec:smooth}

Let us first introduce a simple permutation technique that mitigates the {\em different qubit quality}.

\subsubsection{Permutation Technique.} We already saw in \cref{sec:minimize} that configurations for some quantum circuit $C$ with minimal circuit norm are not unique. Let $M$ be the set of configurations with minimal circuit norm, including all permutations of qubits. Then we may perform measurements for circuits randomly chosen from $M$, see \cref{alg:permutation_general}. This approach averages over the qubit quality, while due to its invariant circuit norm preserving the error rate $\tau(n)$.

\begin{algorithm}[htb]
	\DontPrintSemicolon
	\SetAlgoLined
	\SetKwInOut{Input}{Input}\SetKwInOut{Output}{Output}

	\SetKwComment{COMMENT}{$\triangleright$\ }{}

	Let $M:=\{\text{Configurations of $C$ with minimal circuit norm}\}.$\;
	Evaluate $C$ with configurations chosen randomly from $M$.

	\caption{Permutation Technique.}\label{alg:permutation_general}
\end{algorithm}

\noindent {\em Instantiation of $M$ in our experiments.} First we chose a set of of highest quality qubits $\{i_1, \ldots, i_{2n-1}\}$  together with a starting configuration with minimal circuit norm. Let this be

\[
i_1:x_0 \quad i_2:x_1 \quad i_3:y_1 \quad i_4:x_2 
\quad \ldots \quad i_{2n-2}:x_{n-1} \quad i_{2n-1}:y_{n-1}.
\]
We then chose $b \sim {\cal U}$ and a random permutations $\pi$ on $\{2, \ldots, n-1\}$. This gives us circuit-norm preserving configurations
\[
i_1:x_{b} \quad i_2:x_{1-b} \quad i_3:y_1 \quad i_4:x_{\pi(2)} 
\quad \ldots \quad i_{2n-2}:x_{\pi(n-1)} \quad i_{2n-1}:y_{\pi(n-1)}.
\]

We took 50 circuit-norm preserving configurations, and for each we performed $8192$ measurements on \IBMQ .

The experimental results of our Permuation Technique are illustrated for $n=5$ in \cref{subfig:permut}. In comparison, we have in \cref{subfig:normal} the unsmoothed distribution for $8192$ measurements of a single optimal configuration (as in \cref{fig:simon_experiment}). We already see a significant distribution smoothing, especially vectors with the same Hamming weight obtain similar probabilities. But of course, there is still a clear {\em bias towards 0}, which cannot be mitigated by permutations.

\subsubsection{Double-Flip Technique.} To mitigate the effect that vectors with small Hamming weight are measured more frequently than vectors with large Hamming 
\begin{wrapfigure}{r}{0.5\textwidth}
	\centering
	\vspace*{-1.5em}
	\includegraphics[width=.5\textwidth]{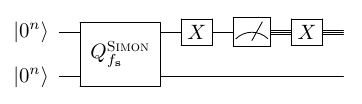}
	\caption{Double-Flip circuit $Q_{DF}$. Triple lines represent classical wires.}\label{circ:double_flip}
	\vspace*{-1em}
\end{wrapfigure}
weight, we flip in Simon's circuit all bits via NOT-gates $X$ before measurement, see~\cref{circ:double_flip}. This flipping inverts the bias towards $0$ that comes from the quantum measurement (not from the previous quantum computation). Since after flipping we  measure the complement, we have to again flip all bits (classically) after measurement and combine them with the original measurements.

\medskip

\noindent {\em Experimental Results and Discussion.} We performed 8192 measurements with circuit $Q_{DF}$ from \cref{circ:double_flip}, the results are illustrated in \cref{subfig:double_flip1}. As expected, we now obtain a bias towards $1$. Hence, in the {\em Double-Flip Technique} we put together the original measurements with $0$-bias from \cref{subfig:normal} and the flipped measurements with $1$-bias from \cref{subfig:double_flip1}, resulting in the smoothed distribution from \cref{subfig:double_flip2}.

From \cref{subfig:double_flip2} we already see that the Double-Flip Technique is quite effective. Moreover, similar to the Permutation Technique, Double-Flip is a general smoothing technique that can be applied for other quantum circuits as well. However, there is also a significant drawback of Double-Flip, since it requires additional (small) quantum circuitry for performing $X$. Thus, as opposed to the Permutation Technique the Double-Flip does not preserve circuit norm. This implies that it slightly increases the error rate $\tau$, as we will see in~\cref{sec:distribution-measure}, where we study more closely the quality of our smoothing techniques.

\begin{figure}[h!]
	\begin{center}
		\begin{adjustbox}{minipage=\linewidth,scale=1}
		\begin{subfigure}[t]{0.49\textwidth}
			\hspace*{-0.2cm}
			\includegraphics[width=\textwidth]{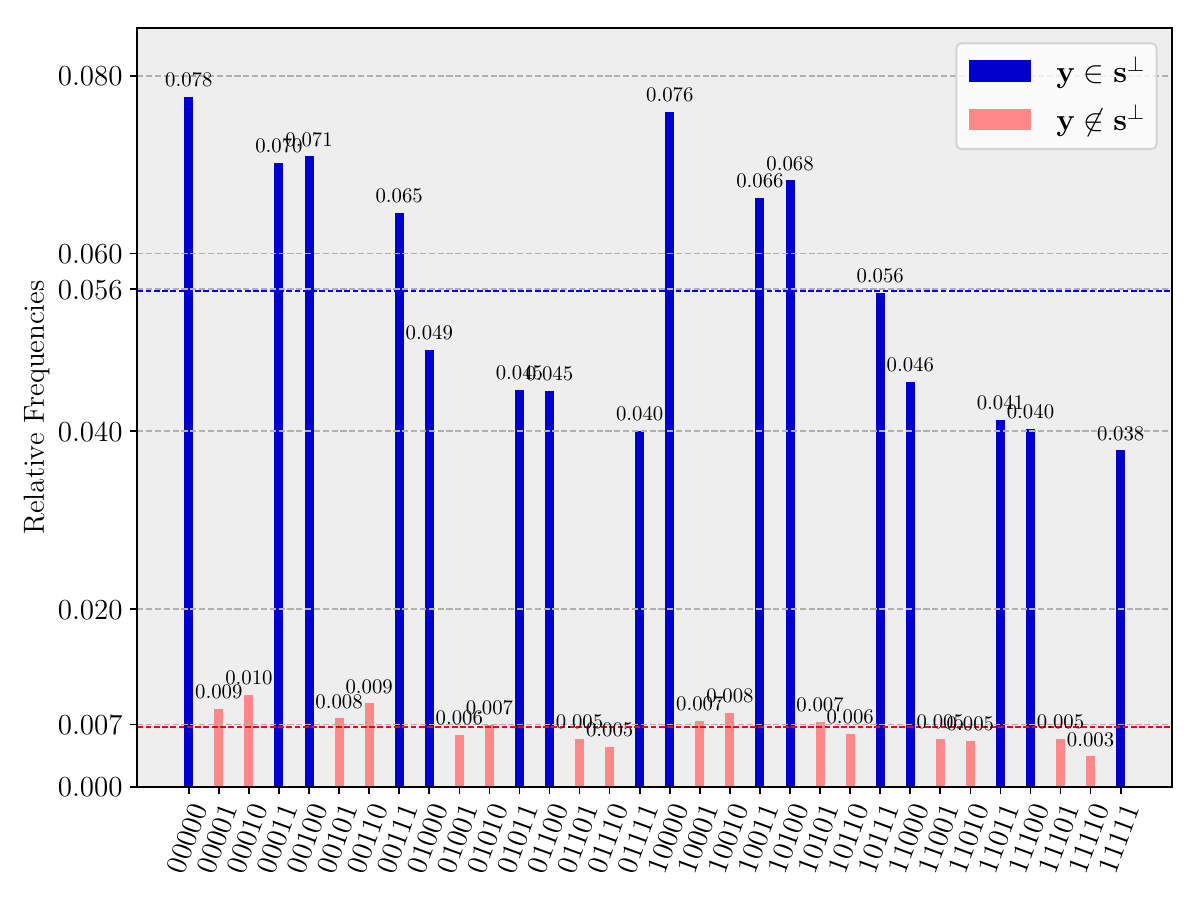}
			\caption{Unsmoothed measurements, n=$5$.\\  \phantom{1}}\label{subfig:normal}
		\end{subfigure}
		\hfill
		\begin{subfigure}[t]{0.49\textwidth}
			\hspace*{-0.2cm}
			\includegraphics[width=\textwidth]{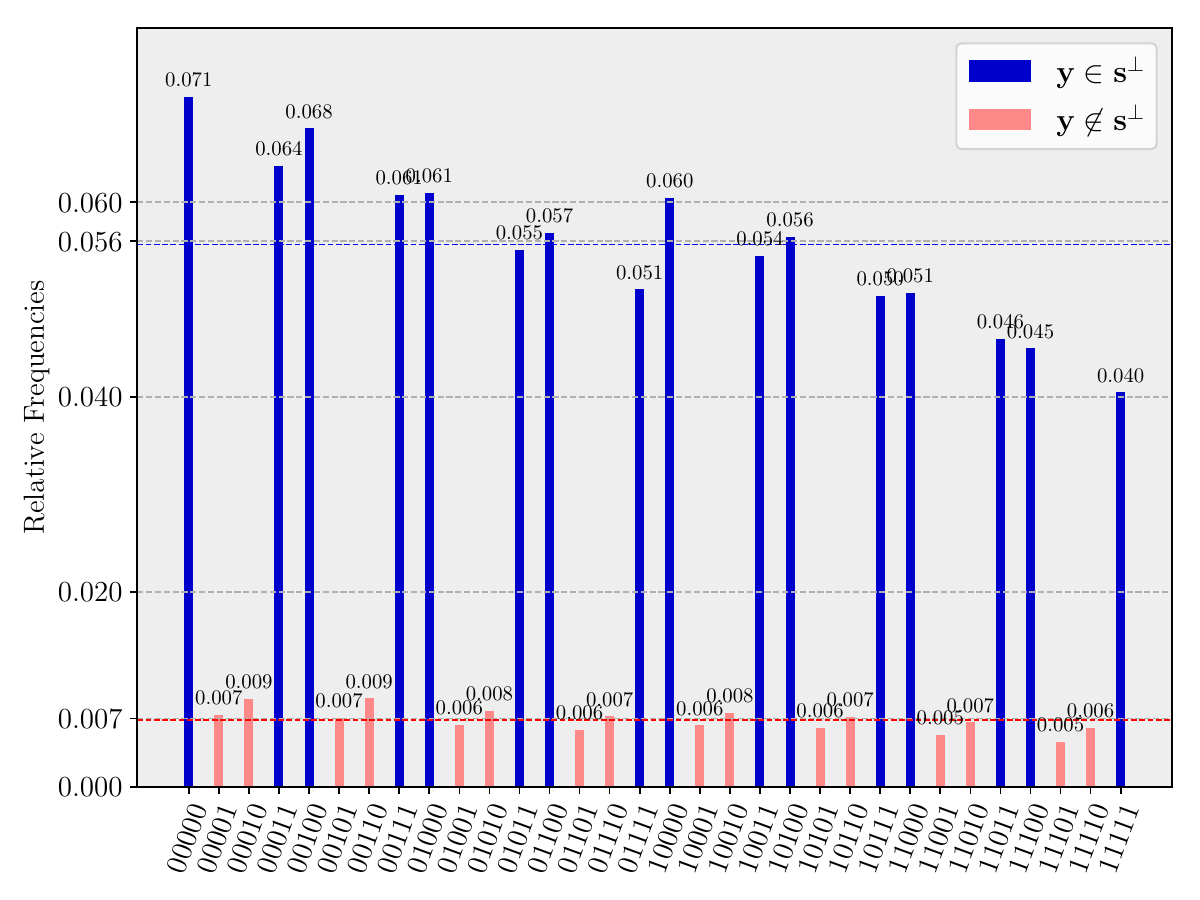}
			\caption{Permutation Technique.\\  \phantom{1}}\label{subfig:permut}
		\end{subfigure}
		\\
		\begin{subfigure}[t]{0.49\textwidth}
			\hspace*{-0.2cm}
			\includegraphics[width=\textwidth]{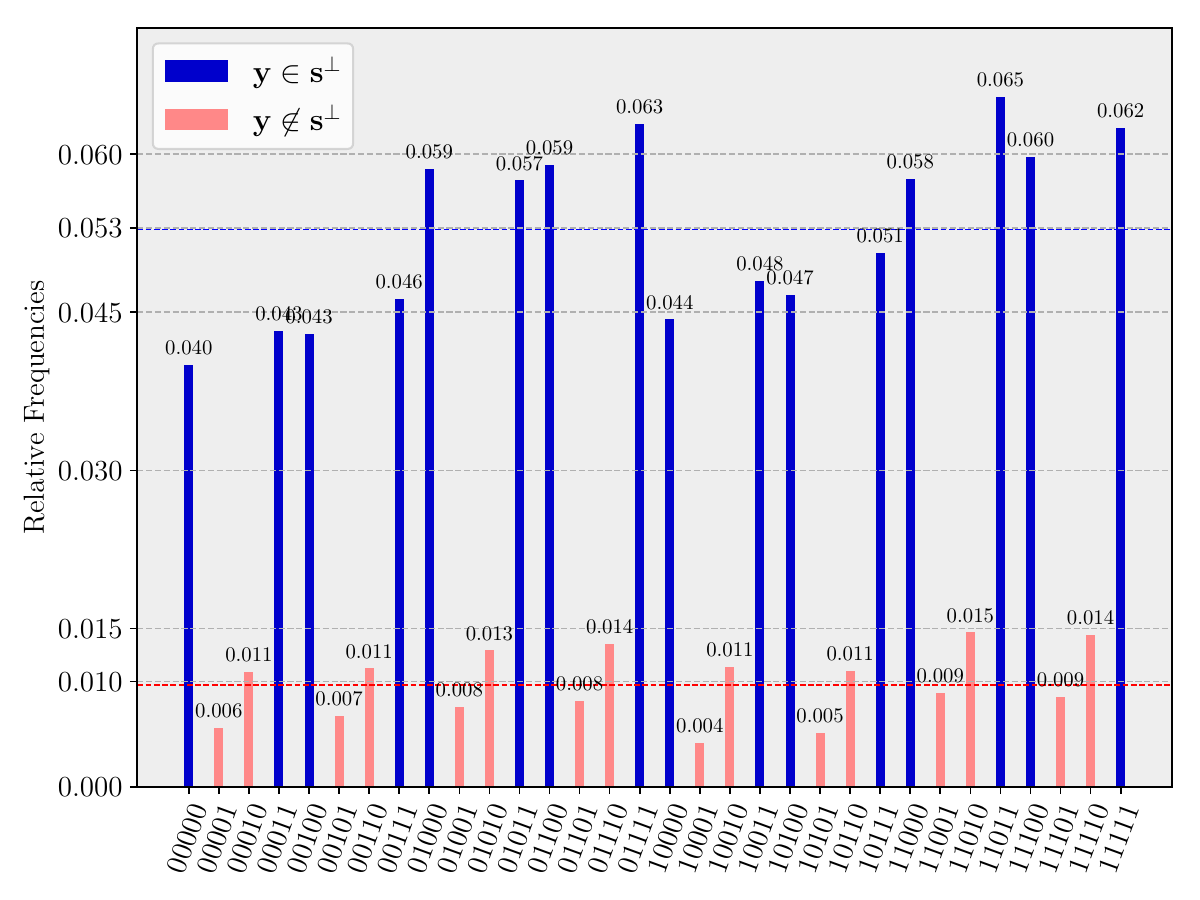}
			\caption{Measurements $Q_{DF}$ (\cref{circ:double_flip}).\\  \phantom{1}}\label{subfig:double_flip1}
		\end{subfigure}
		\hfill
		\begin{subfigure}[t]{0.49\textwidth}
			\hspace*{-0.2cm}
			\includegraphics[width=\textwidth]{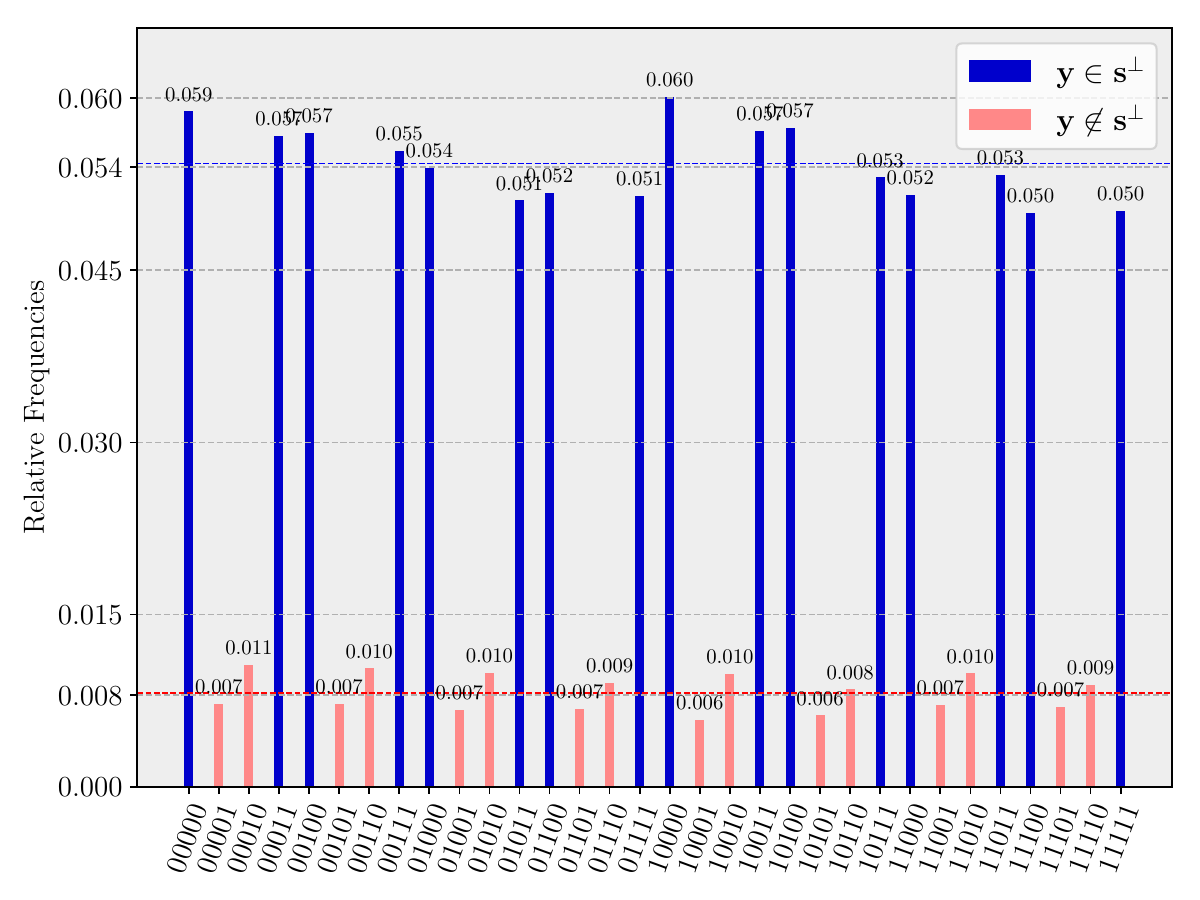}
			\caption{Double-Flip Technique.\\  \phantom{1}}\label{subfig:double_flip2}
		\end{subfigure}
		\\
		\begin{subfigure}[t]{0.49\textwidth}
			\hspace*{-0.2cm}
			\includegraphics[width=\textwidth]{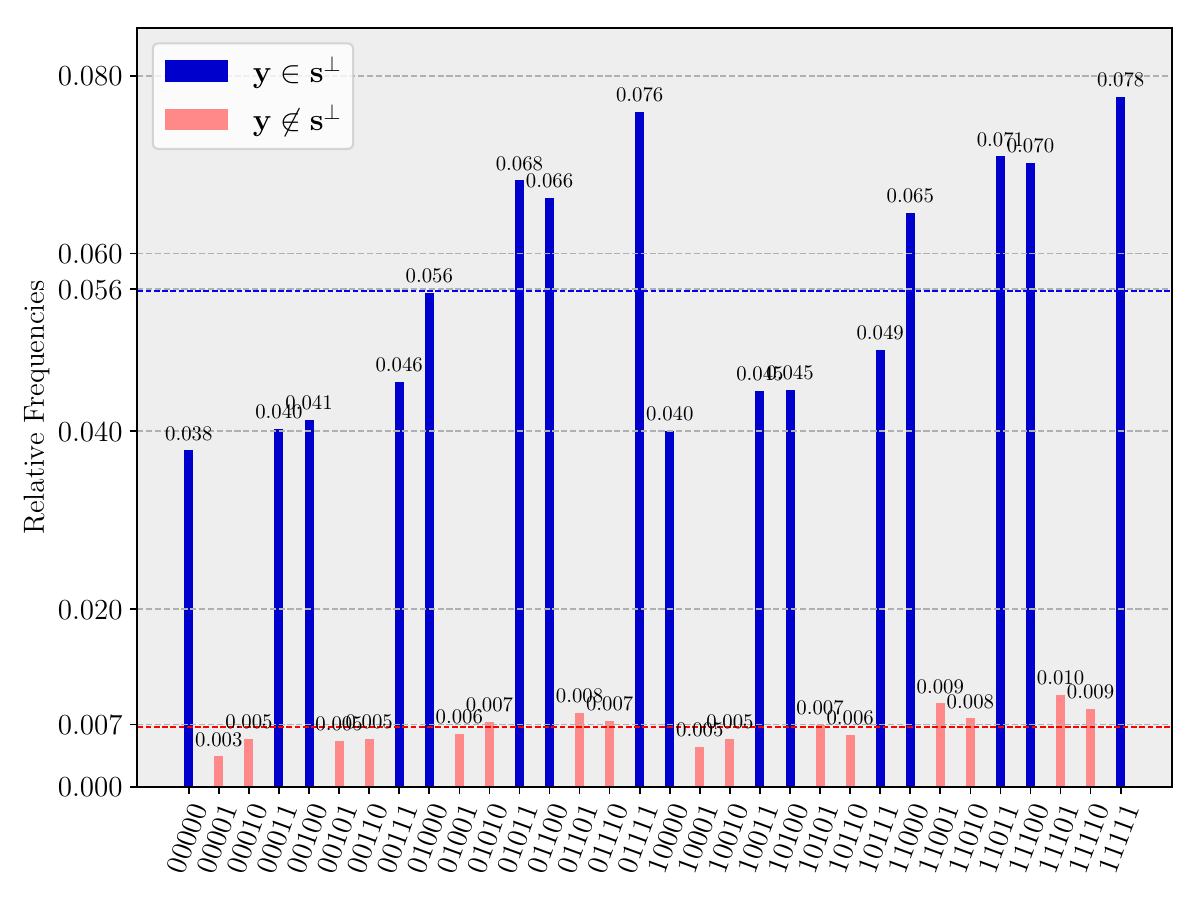}
			\caption{Complemented measurements.\\  \phantom{1}}\label{subfig:hamming1}
		\end{subfigure}
		\hfill
		\begin{subfigure}[t]{0.49\textwidth}
			\hspace*{-0.2cm}
			\includegraphics[width=\textwidth]{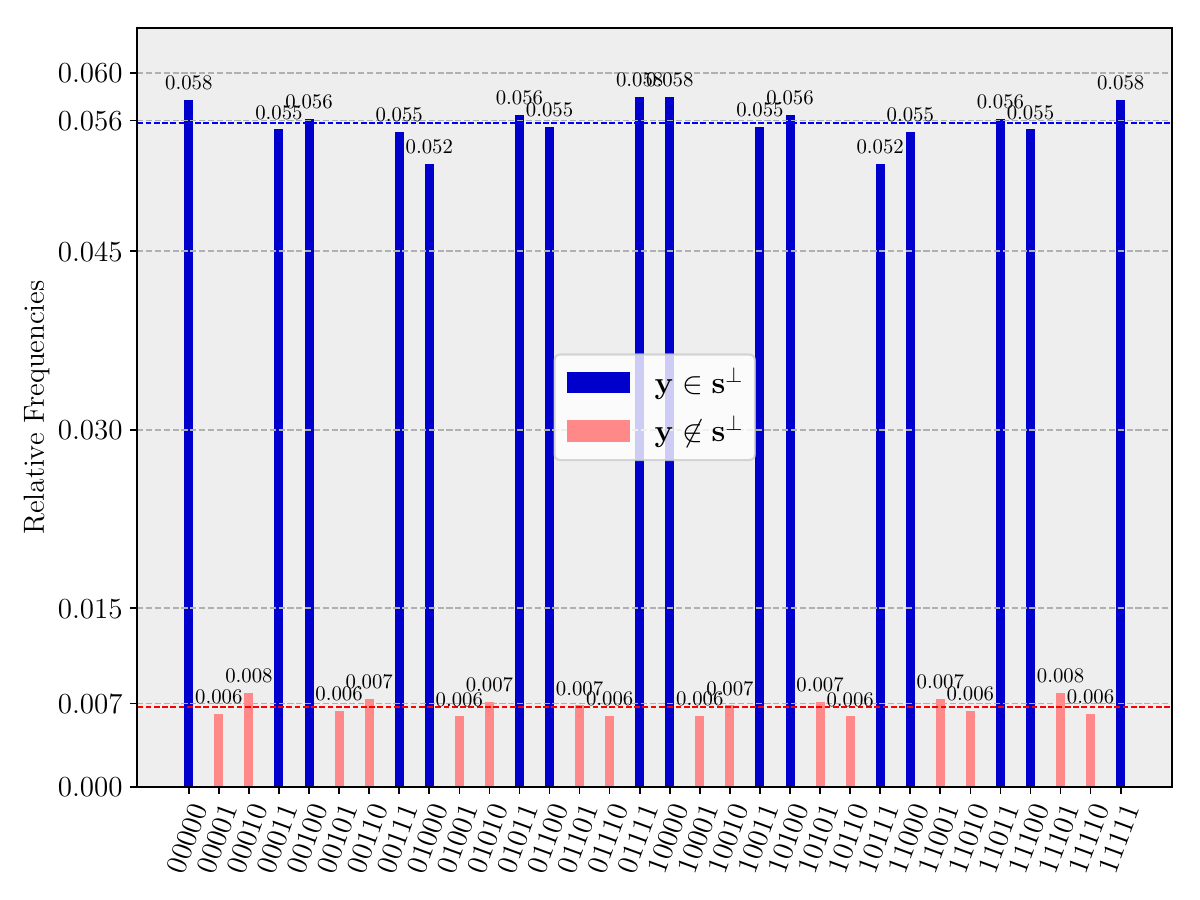}
			\caption{Hamming Technique.\\  \phantom{1}}\label{subfig:hamming2}
		\end{subfigure}
		\end{adjustbox}
	\end{center}
	\vspace*{-1em}
	\caption{Smoothed \IBMQ\ measurements.\\ \hspace*{1cm}}\label{fig:smooth}\vspace*{-3em}
\end{figure}

\subsubsection{Hamming Technique.} The Hamming Technique is similar to the Double-Flip Technique, but as opposed to Double-Flip Hamming is specific to Simon-type problems and a purely classical post-processing of data without adding any additional circuitry. 

Let $Q \subseteq \F_2^n$ be a multiset of quantum measurements, e.g. the set of $8192$ measurements from \cref{subfig:normal}. Then consider the complementary multiset
\[
  \bar Q = \{\vec q + 1^n \mid \vec q \in Q\},
\]
where we flip all bits. Let $\vec q \in Q \cap \vec s^{\perp}$, i.e. $\vec q$ is a measurement in the subspace orthogonal to $\vec s$. By complementing $Q$ we want to preserve orthogonality, i.e. we want to have $\vec q + 1^n \in \vec s^{\perp}$ which is true iff $1^n \in \vec s^{\perp}$ by the subspace structure.

Thus, complementation preserves orthogonality iff $1^n \in \vec s^{\perp}$, which is in turn equivalent to even Hamming weight $h(\vec s)$. Similar to Double-Flip, in the {\em Hamming Technique} we combine both measurements $Q \cup \bar Q$. The Hamming Technique mitigates the effect that for each $\vec q \in Q$ with large frequency (due to the $0$-bias) we also obtain $\vec q + 1^n$ with small frequency (due to the $0$-bias), and vice versa. Thus, averaging both frequencies should smooth our distribution closer to  uniformity.

What happens if $1^n \notin \vec s^{\perp}$? We want to add some vector $\vec v \in \F_2^n$ with Hamming weight as large as possible. It is not hard to see that there always exists some $\vec v \in \vec s^{\perp}$ with $h(\vec v) \geq n-1$. Thus, we can simply try all $n+1$ possible vectors.

\medskip

\noindent {\em Experimental Results.} Since our instantiation of $f_{\vec s}$ from \cref{sec:fchoice} uses even-weight periods $\vec s$, we can use the multiset $\bar Q$ (with $1^n$), which was done in \cref{subfig:hamming1} and is a direct mirroring of $Q$ in \cref{subfig:normal}. The multiset of measurement $Q \cup \bar Q$ is then depicted in \cref{subfig:hamming2}.

In comparison with Double-Flip from \cref{subfig:double_flip2}, we see that the Hamming technique provides in \cref{subfig:hamming2} a distribution which is closer to the uniform distribution among orthogonal and non-orthonal vectors. Thus, for our experimental data one should prefer the Hamming technique over Double-Flip.

\vspace*{-1em}
\subsubsection{Combination of techniques.} The same preference can be observed when we combine the Permutation technique with either Double-Flip (see \cref{subfig:combined1}) or with Hamming (see \cref{subfig:combined2}).

\begin{figure}[h!]
	\vspace*{-1em}
	\begin{center}
		\begin{adjustbox}{minipage=\linewidth,scale=1}
		\begin{subfigure}[t]{0.49\textwidth}
			\hspace*{-0.2cm}
			\includegraphics[width=\textwidth]{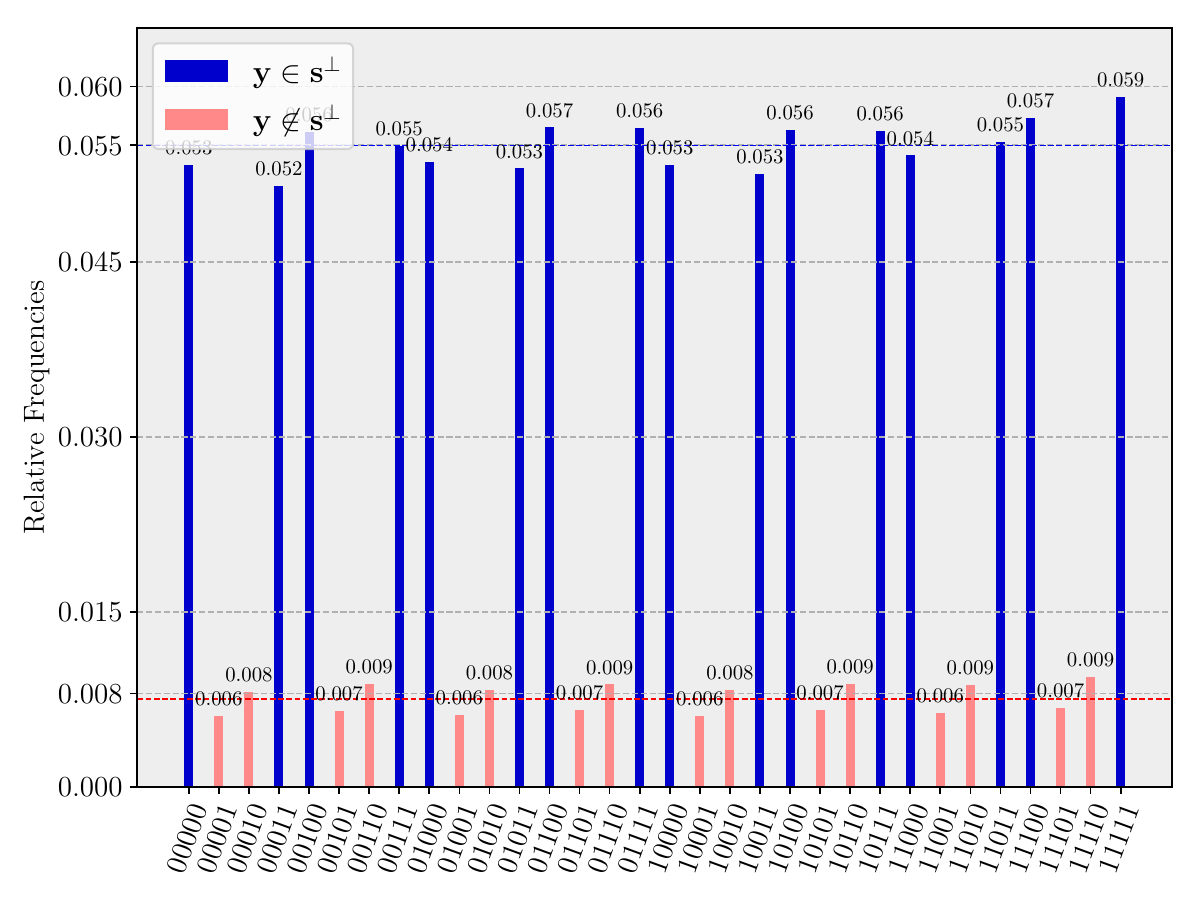}
			\caption{Combined Permutation/Double-Flip}\label{subfig:combined1}
		\end{subfigure}
		\hfill
		\begin{subfigure}[t]{0.49\textwidth}
			\hspace*{-0.2cm}
			\includegraphics[width=\textwidth]{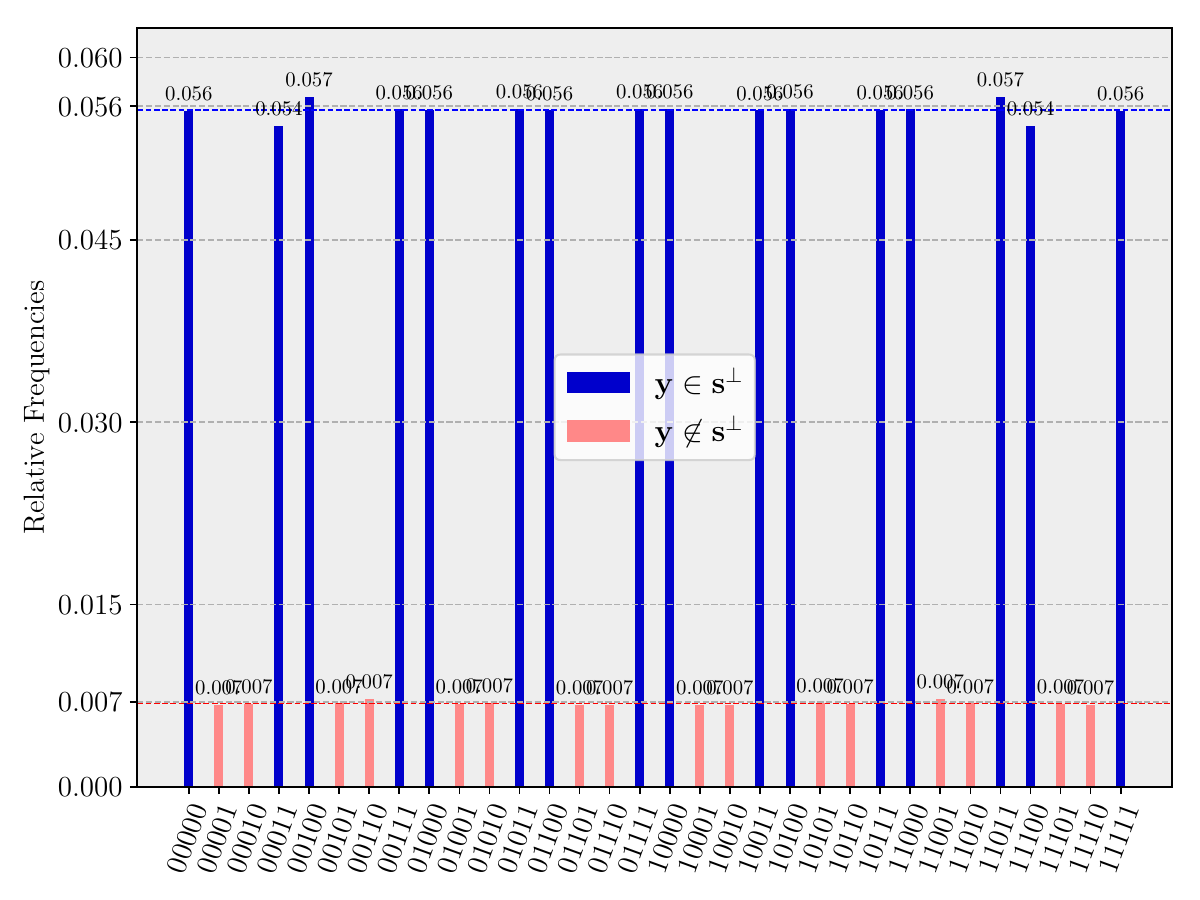}
			\caption{Combined Permutation/Hamming}\label{subfig:combined2}
		\end{subfigure}
		\end{adjustbox}
	\end{center}
	\vspace*{-1em}
	\caption{Smoothing using a combination of techniques.}\label{fig:com_smooth}\vspace*{-0.5em}
\end{figure}

The combination Permutation/Hamming seems to outperform Permutation/ Double-Flip, and Permutation/Hamming almost optimally follows our LSN Error Model from \cref{def:error_model}.

\subsection{Quality Measures Statistics.}
\label{sec:distribution-measure}

Let us introduce a well-known statistical distance that quantitatively measures the effectiveness of our smoothing techniques. Recall that we require error distributions close to our LSN Error Model, in order to justify the proper use of LPN solvers in subsequent sections.

The Kullback-Leibler divergence describes the loss of information when going from a distribution $P$ -- e.g. our LSN Error Model distribution -- to another distribution $Q$ -- e.g. our smoothed \IBMQ\ measurements.

\begin{definition}[Kullback–Leibler divergence (KL)]\label{def:KL}
The Kullback-Leibler divergence  of two probability distributions $P$ towards $Q$ on $\F_2^n$ is
	\[D_\text{KL}(P||Q):=\sum_{y\in\F_2^n}P(Y)\log\left(\frac{P(y)}{Q(y)}\right)\;.\]
\end{definition}

We compute KL and the error rate $\tau$ on the data from  \cref{fig:smooth,fig:com_smooth}. The results are given in \cref{tab:summary}.

\begin{table}[htb]
	\centering
	\begin{tabular}{|*{3}{c|}}\hline
		\backslashbox{Smoothing}{Measure} 	&\text{KL}& $\tau$	\\\hline
		None							&$ \ 0.04644 \ $& $ \ 0.10730 \ $  \\
		Permutation						&$0.01596$&$0.10954$  \\
		Double-Flip						&$0.00600$&$0.13104$  \\
		Permutation/Double-Flip			&$0.00297$&$0.12044$  \\
		Hamming							&$0.00139$&$0.10730$  \\
		Permutation/Hamming				&$0.00011$&$0.10954$  \\\hline
	\end{tabular}
	\vspace*{0.1cm}
	\caption{Kullback–Leibler applied to our Smoothing Techniques.
	}\label{tab:summary}
\end{table}
\vspace*{-1em}

As we would expect for KL, Hamming is more effective than Double-Flip. Also as predicted, Double-Flip increases the error rate $\tau$, whereas the other techniques leave $\tau$ (basically) unchanged. In particular, Hamming leaves $\tau$ unchanged, since it is only a classical post-processing of our quantum data. We have already seen qualitatively in \cref{fig:com_smooth} that the combination Permutation/Hamming performs best. This is supported also quantitatively in \cref{tab:summary}: KL is very close to zero, indicating that via Permutation/Hamming smoothed \IBMQ\ quantum measurements almost perfectly agree with the LSN Error Model.

The results of applying Permutation/Hamming to all $n=2, \ldots, 7$ are depicted in \cref{fig:simon_smooth_experiment}.

\begin{figure}[H]
	\begin{center}
		\begin{adjustbox}{minipage=\linewidth,scale=1}
		\begin{subfigure}[t]{0.49\textwidth}
			\hspace*{-0.2cm}
			\includegraphics[width=\textwidth]{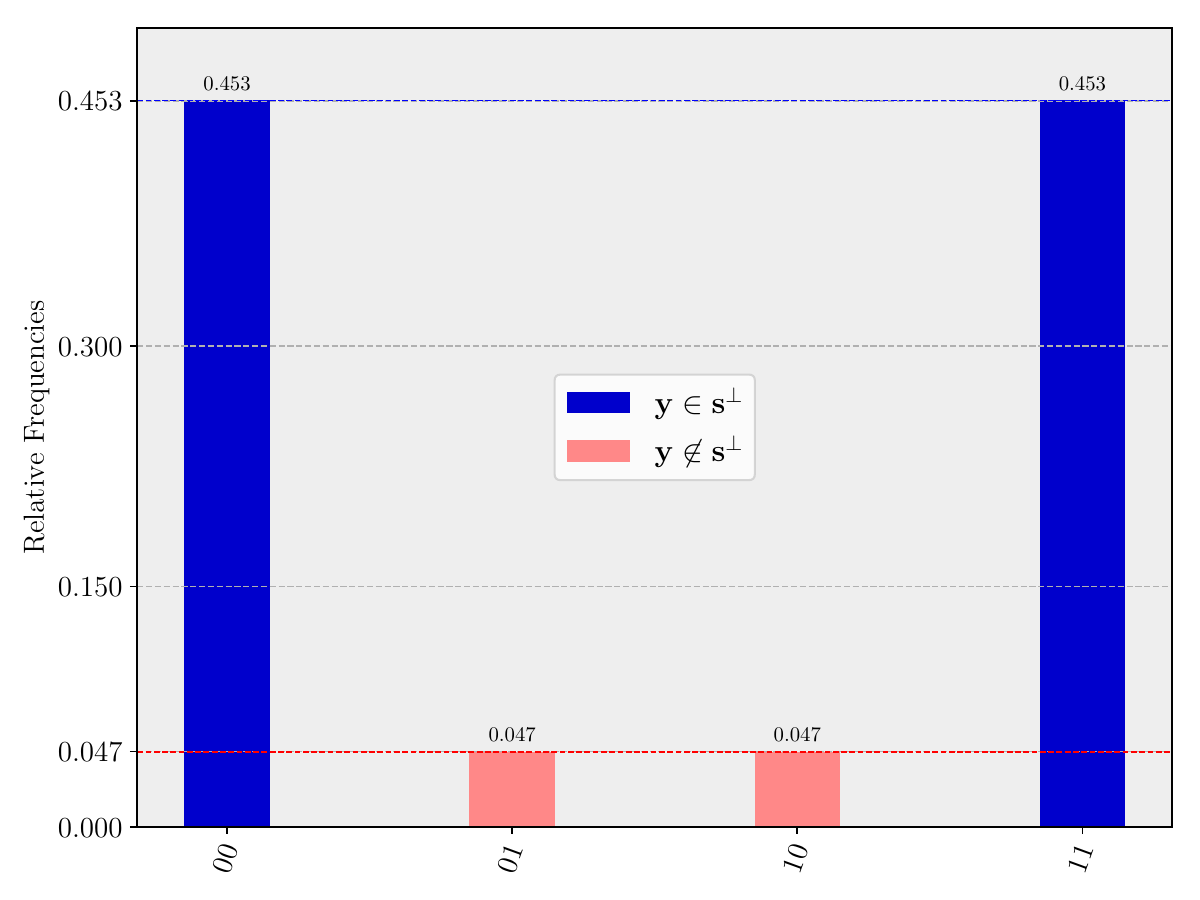}
			\caption{$n=2$, $\tau = 0.09347$, $\text{KL} = 0.00000$.\\  \phantom{1}}
		\end{subfigure}
		\hfill
		\begin{subfigure}[t]{0.49\textwidth}
			\hspace*{-0.2cm}
			\includegraphics[width=\textwidth]{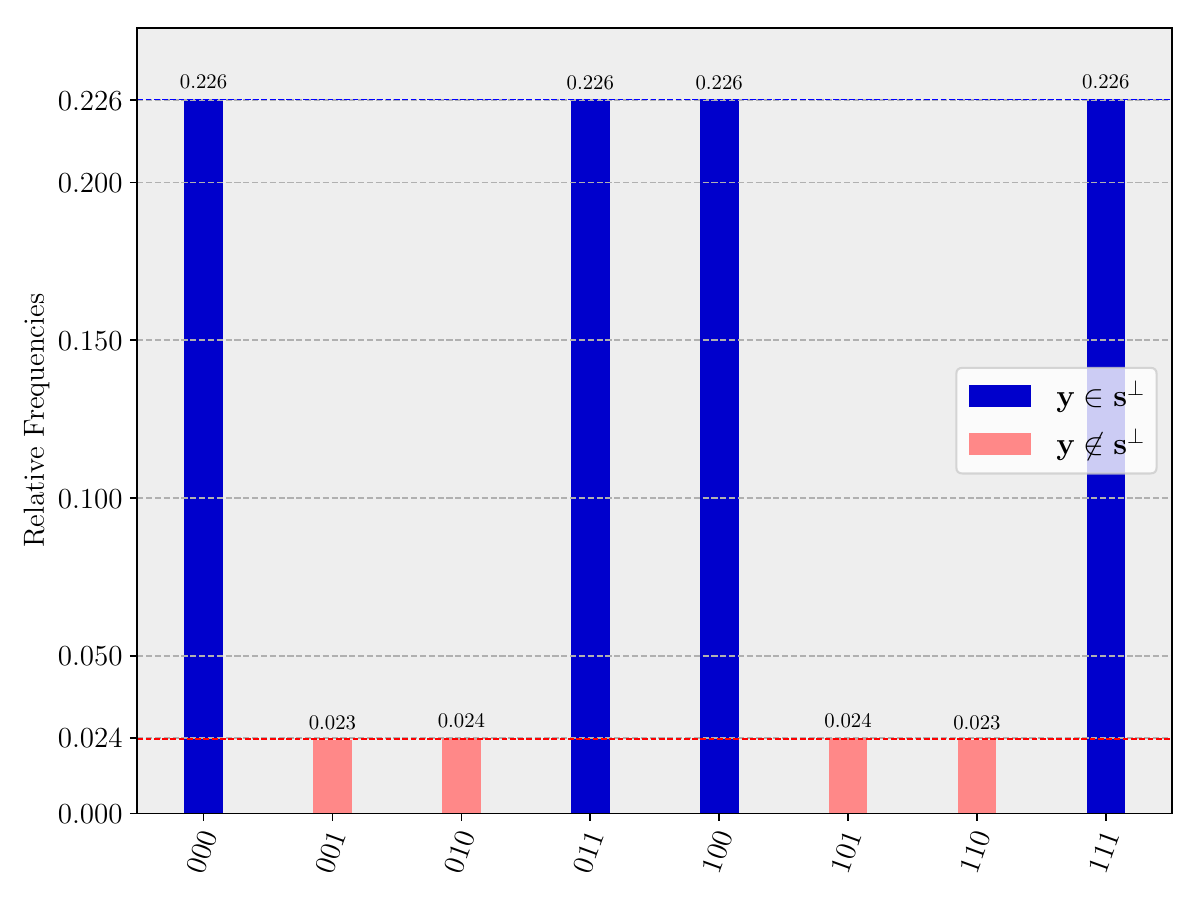}
			\caption{$n=3$, $\tau = 0.09479$, $\text{KL} = 0.00002$.\\  \phantom{1}}
		\end{subfigure}
		\\
		\begin{subfigure}[t]{0.49\textwidth}
			\hspace*{-0.2cm}
			\includegraphics[width=\textwidth]{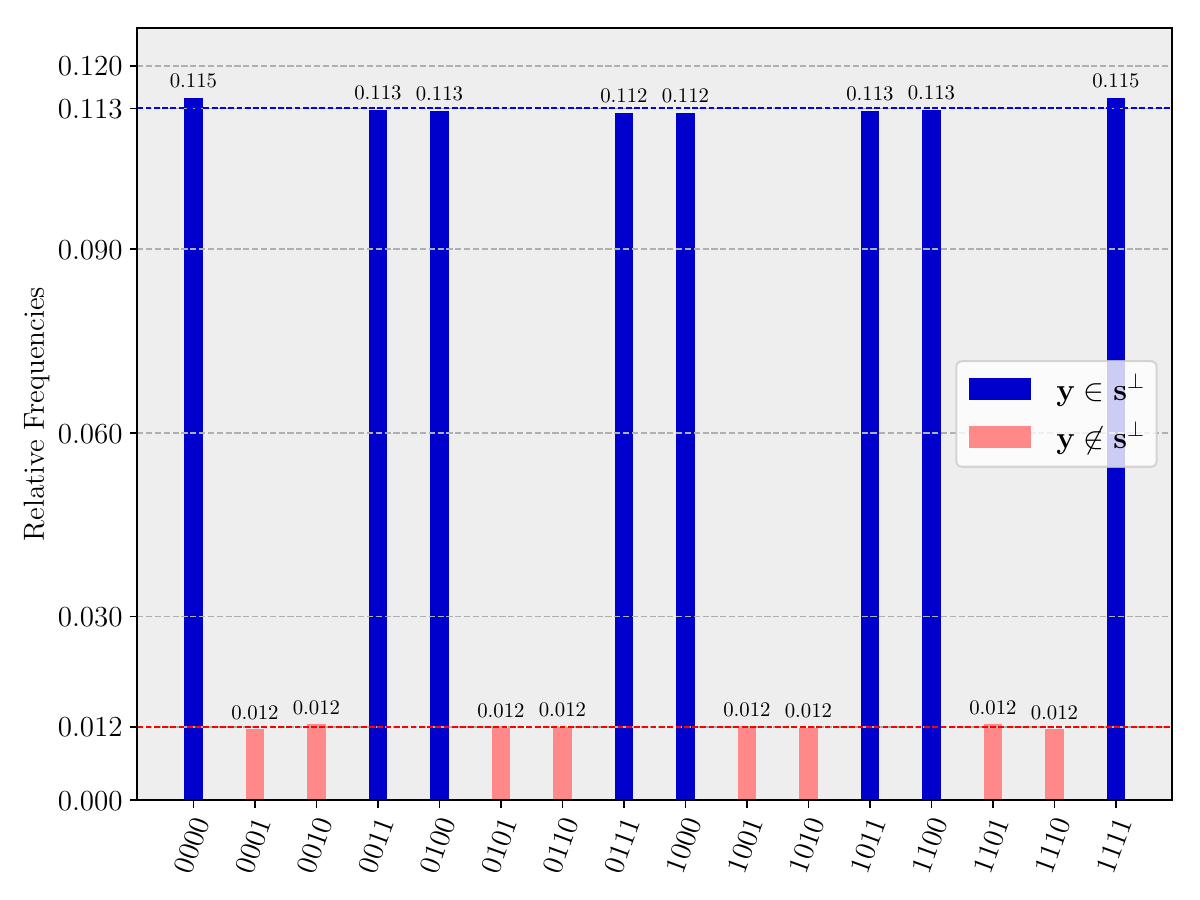}
			\caption{$n=4$, $\tau = 0.09546$, $\text{KL} = 0.00009$.\\  \phantom{1}}
		\end{subfigure}
		\hfill
		\begin{subfigure}[t]{0.49\textwidth}
			\hspace*{-0.2cm}
			\includegraphics[width=\textwidth]{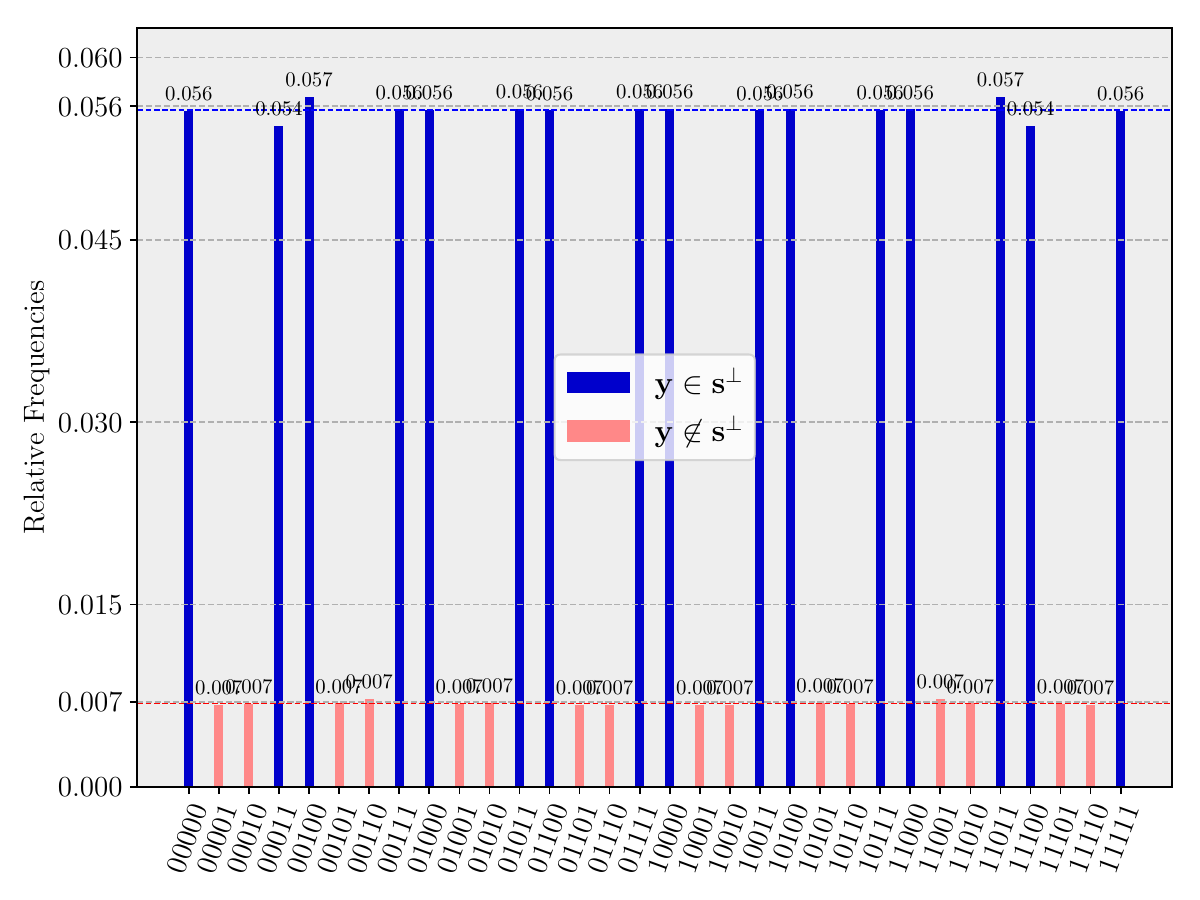}
			\caption{$n=5$, $\tau = 0.10954$, $\text{KL} = 0.00011$.\\  \phantom{1}}
		\end{subfigure}
		\\
		\begin{subfigure}[t]{0.49\textwidth}
			\hspace*{-0.2cm}
			\includegraphics[width=\textwidth]{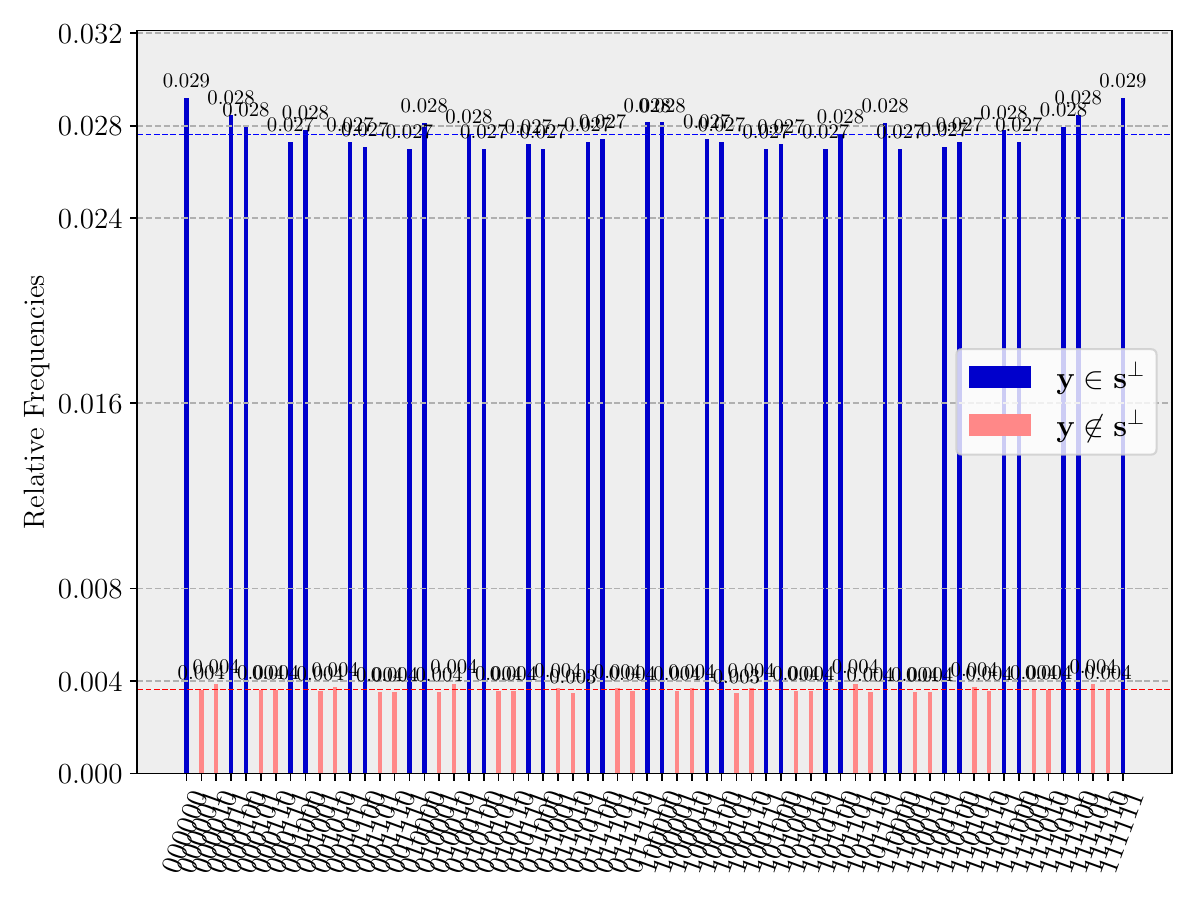}
			\caption{$n=6$, $\tau = 0.11602$, $\text{KL} = 0.00038$.}
		\end{subfigure}
		\hfill
		\begin{subfigure}[t]{0.49\textwidth}
			\hspace*{-0.2cm}
			\includegraphics[width=\textwidth]{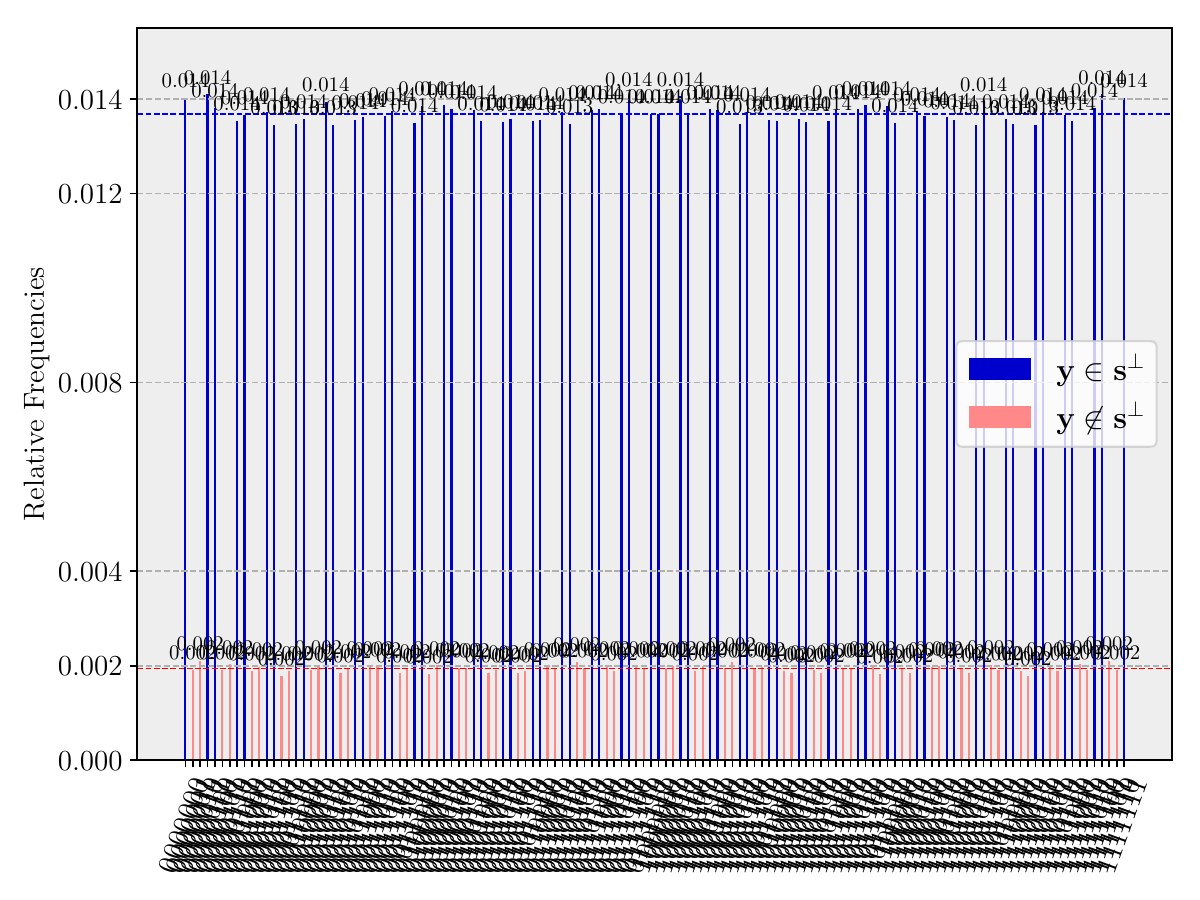}
			\caption{$n=7$, $\tau = 0.12398$, $\text{KL} = 0.00022$.}
		\end{subfigure}
		\end{adjustbox}
	\end{center}
\vspace*{-1em}
	\caption{Via Permutation/Hamming Technique smoothed \IBMQ\ measurements for $n=2, \ldots, 7$. $\text{KL}$ is the Kullback-Leibler divergence to the LSN Error Model distribution.}\label{fig:simon_smooth_experiment}
\end{figure}

\section{\LSN\ is Polynomial Time Equivalent to \LPN}
\label{sec:reduction}

In the previous section, we smoothed our \IBMQ\ experiments to the LSN Error Model (\cref{def:error_model}). Recall that the LSN Error Model states that with probability $\tau$ we measure in the quantum circuit $Q^{\Simon}_{f_{\vec s}}$ some uniformly distributed $\vec y \in \F_2^n \setminus \sorth$. The question is now whether such erroneous $\vec y$ as in our error model can easily be handled, i.e. whether \LSN\ can be efficiently solved.

In this section, we answer this question in the negative. Namely, we show that solving LSN is tightly as hard as solving the well-studied \LPN\  problem, which is supposed to be hard even on quantum computers.

\begin{definition}[\LPN-Problem]\label{def:LPN}
	Let $\vec s \in \F_2^n \setminus \Zero$ be chosen uniformly at random, and let $\tau\in[0,\frac{1}{2})$.
	In the {\em Learning Parity with Noise} problem, denoted $\LPN_{n,\tau}$, one obtains access to an oracle ${\cal O}_{\LPN}(\vec s)$ that provides samples $(\vec a,\langle \vec a,\vec s \rangle + \epsilon)$, where $\vec a \sim {\cal U}_n$ and $\epsilon \sim \mathrm{Ber}_{\tau}$. The goal is to compute $\vec s$.
\end{definition}

\cref{def:LPN} explicitly excludes $\vec s = \vec 0$ in \LPN. Notice that the case $\vec s = \vec 0$ implies that the \LPN\ oracle has distribution $U_n \times \textrm{Ber}_{\tau}$, whereas in the case $\vec s \not= \vec 0$ we have $\Pr_{\vec a}[\langle \vec a , \vec s \rangle=0] = \frac 1 2$ and therefore $\Pr_{\vec a}[\langle \vec a , \vec s \rangle + \epsilon=0] = \frac 1 2$. Hence, for $\vec s \not= \vec 0$ the \LPN\ samples  have distribution $U_n \times U$. This allows us to easily distinguish both cases by a majority test, whenever $\tau$ is polynomially bounded away from $\frac 1 2$. In conclusion, $\vec s= \vec 0$ is not a hard case for \LPN\ and may wlog be excluded.

Let us now define the related {\em Learning Simon with Noise} problem that reflects the LSN Error Model.

\begin{definition}[\LSN-Problem]\label{def:LSN}
	Let $\vec s \in \F_2^n \setminus \Zero$ be chosen uniformly at random, and let $\tau\in[0,\frac{1}{2})$.
	In the {\em Learning Simon with Noise} problem, denoted $\LSN_{n,\tau}$, one obtains access to an oracle ${\cal O}_{\LSN}(\vec s)$ that provides samples $\vec y$, where $\vec y \in \F_2^n$ is distributed as in \cref{def:error_model}, i.e.
\[
\Pr[\vec y] =
\begin{cases}
\frac{1-\tau}{2^{n-1}} & \textrm{, if }  y \in \vec s^{\perp} \\
\frac{\tau}{2^{n-1}} & \textrm{, else}
\end{cases}\; \textrm{ and therefore } \Pr[ \langle \vec y , \vec s \rangle = 0 ] = 1-\tau.
\]
The goal is to compute $\vec s$.
\end{definition}

In the following we prove that $\LSN_{n,\tau}$ is polynomial time equivalent to $\LPN_{n,\tau}$ by showing that we can perfectly mutually simulate ${\cal O}_{\LPN}(\vec s)$ and ${\cal O}_{\LSN}(\vec s)$. The purpose of excluding $\vec s \neq \vec 0$ from $\LPN_{n,\tau}$ is to guarantee in the reduction non-trivial periods $\vec s \neq \vec 0$ in $\LSN_{n,\tau}$.

\begin{theorem}[Equivalence of \LPN\ and \LSN]\label{theo:equivalence}
Let ${\cal A}$ be an algorithm that solves $\LPN_{n,\tau}$ (respectively $\LSN_{n,\tau}$) using $m$ oracle queries in time $T$ with success probability $\epsilon_{\cal A}$. Then there exists an algorithm  ${\cal B}$ that solves  $\LSN_{n,\tau}$ (respectively $\LPN_{n,\tau}$) using $m$ oracle queries in time $T$ with success probability $\epsilon_{\cal B} \geq \frac {\epsilon_{\cal A}} 2$.
\end{theorem}


\begin{proof}
Assume that we want to solve $\LSN$ via an algorithm ${\cal A}_{\LPN}$ with success probability $\epsilon_{\cal A}$ as in \cref{alg:LPNtoLSN}.

		\begin{algorithm}[htb]
			\DontPrintSemicolon
			\SetAlgoLined
			\SetKwInOut{Input}{Input}\SetKwInOut{Output}{Output}

			\SetKwComment{COMMENT}{$\triangleright$\ }{}
			\Input{
				$n, \tau, {\cal O}_{\LSN}(\vec s), m$
			}
			\Output{
				$\vec s$
			}

			Choose $\vec z \sim {\cal U}_n$. \label{line:z}\;
			\For{$i =1$ to $m$}{
				Set $\vec y_i \leftarrow {\cal O}_{\LSN}(\vec s)$. \;
				Choose $ b_i \sim {\cal U}$. \;
			}
			$\vec s \leftarrow {\cal A}_{\LPN}(n, \tau, (\vec y_1 + b_1\vec z, b_1), \ldots, (\vec y_m + b_m\vec z, b_m))$
			\caption{\textsc{\LPN\ $\Rightarrow$ \LSN}}\label{alg:LPNtoLSN}
		\end{algorithm}

We show in the following that \cref{alg:LPNtoLSN} perfectly simulates the oracle ${\cal O}_{\LPN}(\vec s)$ via ${\cal O}_{\LSN}(\vec s)$ if the vector $\vec z \sim {\cal U}_n$ chosen in \cref{line:z} satisfies $\langle \vec z, \vec s \rangle =1$. Since $\vec s \not= \vec 0$, we have $\Pr_{\vec z}[\langle \vec z, \vec s \rangle =1] = \frac 1 2$. Therefore Algorithm~\ref{alg:LPNtoLSN} succeeds with probability
\[
  \epsilon_{\cal B} \geq \Pr_{\vec z}[\langle \vec z, \vec s \rangle = 1 \cap {\cal A} \textrm{ outputs } \vec s] = \frac {\epsilon_{\cal A}} 2.
\]

Let us now show correctness of \cref{alg:LPNtoLSN}. We first show that the constructed \LPN\ samples $(\vec y + b \vec z, b)$ have the correct distribution. Let $\epsilon = \langle \vec y + b\vec z, \vec s \rangle + b$. Since $\langle \vec z, \vec s \rangle =1$, we have
\[
  \Pr_{\vec y}[\epsilon =1] = \Pr_{\vec y}[\langle \vec y + b \vec z, \vec s \rangle + b = 1] = \Pr_{\vec y}[\langle \vec y, \vec s \rangle + b \langle \vec z, \vec s \rangle + b = 1] = \Pr_{\vec y}[\langle \vec y , \vec s \rangle = 1] = \tau.
\]
It remains to show that $\vec y + b \vec z$ is uniformly distributed. To this end, we show that
\[ p_0 = \Pr_{\vec y, b}[\vec y + b \vec z \mid \langle \vec y, \vec s \rangle = 0] = \frac 1 {2^n}.
\]
Analogous, it follows that  $p_1=\Pr_{\vec y, b}[\vec y + b \vec z \mid \langle \vec y, \vec s \rangle = 1] = \frac 1 {2^n}$. From both statements we obtain
\[
  \Pr_{\vec y, b}[\vec y + b \vec z] = \Pr_{\vec y}[\langle \vec y, \vec s \rangle = 0] \cdot p_0 + \Pr_{\vec y}[\langle \vec y, \vec s \rangle = 1] \cdot p_1 = \frac{1-\tau}{2^n} + \frac{\tau}{2^n} = \frac 1 {2^n},
\]
as desired. It remains to show that
\begin{align*}
p_0 & = \Pr_{\vec y, b}[\vec y + b \vec z \mid \langle \vec y, \vec s \rangle = 0] \\
 & = \Pr_b[b=0] \cdot  \Pr_{\vec y}[\vec y \mid \langle \vec y, \vec s \rangle = 0] + \Pr_b[b=1] \cdot  \Pr_{\vec a}[\vec y + \vec z \mid \langle \vec y, \vec s \rangle = 0] \\
& = \frac 1 2 \left( \frac{1-\tau}{2^{n-1}} + \frac{\tau}{2^{n-1}} \right) = \frac{1}{2^n}.
\end{align*}
This completes the analysis of Algorithm~\ref{alg:LPNtoLSN}.

		\begin{algorithm}[htb]
			\DontPrintSemicolon
			\SetAlgoLined
			\SetKwInOut{Input}{Input}\SetKwInOut{Output}{Output}

			\SetKwComment{COMMENT}{$\triangleright$\ }{}
			\Input{
				$n, \tau, {\cal O}_{\LPN}(\vec s), m$
			}
			\Output{
				$\vec s$
			}

			Choose $\vec z \sim {\cal U}_n$. \;
			\For{$i =1$ to $m$}{
				Set $(\vec a_i, b_i) \leftarrow {\cal O}_{\LPN}(\vec s)$. \;
			}
			$\vec s \leftarrow {\cal A}_{\LSN}(n, \tau, \vec a_1 + b_1\vec z, \ldots, \vec a_m + b_m\vec z)$
			\caption{\textsc{\LSN\ $\Rightarrow$ \LPN}\label{alg:LSNtoLPN}}
		\end{algorithm}

For Algorithm~\ref{alg:LSNtoLPN} we conclude the success probability analogous to the reasoning for Algorithm~\ref{alg:LPNtoLSN}, i.e. we succeed when $\langle \vec z , \vec s \rangle =1$ and ${\cal A}_{\LSN}$ succeeds. So let us assume in the following correctness analysis that we are in the case $\langle \vec z , \vec s \rangle =1$. This implies for the constructed \LSN\ samples $\vec a + b \vec z$ that
\[
\langle \vec a + b \vec z, \vec s \rangle = 0 \Leftrightarrow \langle \vec a , \vec s \rangle + b \langle \vec z , \vec s \rangle = 0 \Leftrightarrow \langle \vec a, \vec s \rangle = b.
\]
Let $\epsilon= \langle \vec a, \vec s \rangle + b$. It follows that
\[
   \Pr_{\vec a, b}[\langle \vec a + b \vec z, \vec s \rangle = 0] = \Pr_{\vec a, b}[ \langle \vec a, \vec s \rangle = b] = \Pr_{\vec a, b}[\epsilon=0] =  1- \tau.
\]
We also have to show that we obtain a uniform distribution among all $\vec a + b \vec z \in \vec s^{\perp}$. This follows from
\begin{align*}
 & \ \Pr_{\vec a, b}[\vec a + b \vec z \mid \langle \vec a + b \vec z, \vec s \rangle = 0 ] = \Pr_{\vec a, b}[\vec a + b \vec z \mid \langle \vec a , \vec s \rangle = b ] \\
  =  & \ \Pr_{\vec a}[\langle \vec a, \vec s \rangle=0] \cdot \Pr_{\vec a, b}[\vec a + b \vec z \mid \langle \vec a , \vec s \rangle = b = 0 ] \ + \\
  & \  \Pr_{\vec a}[\langle \vec a, \vec s \rangle=1] \cdot \Pr_{\vec a, b}[\vec a + b \vec z \mid \langle \vec a , \vec s \rangle = b = 1 ] \\
   = &  \ \frac 1 2 \cdot \Pr_{\vec a}[\vec a \mid \langle \vec a , \vec s \rangle  = 0 ] + \frac 1 2 \cdot \Pr_{\vec a}[\vec a + \vec z \mid \langle \vec a , \vec s \rangle = 1 ] \\
   =  & \ \frac 1 2 \cdot \frac{1}{2^{n-1}} +\frac 1 2 \cdot \frac{1}{2^{n-1}} =  \frac{1}{2^{n-1}}.
\end{align*}

Analogous, we can show that we obtain a uniform distribution among all $\vec a + b \vec z \in\F_2^n\setminus\sorth$. This proves that we perfectly simulate $\LSN$-samples via ${\cal O}_{\LPN}$, and thus shows correctness of \cref{alg:LSNtoLPN}.
\end{proof}

\cref{theo:equivalence} shows that under the $\LPN$ assumption we cannot expect to solve $\LSN$
in polynomial time. However, it does not exclude that quantum measurements that lead to an $\LSN$ distribution are still useful in the sense that they help us to solve period finding faster than on classical computers.
In the following section, we show that $\LSN$ distributed quantum outputs indeed lead to speedups even for large error rates $\tau$.

%
%

\section{Theoretical Error Handling for Simon's Algorithm}
\label{sec:correct}

It is well-known~\cite{MontanaroW16} that period finding for $n$-bit Simon  functions classically requires  time $\Omega(2^{\frac n 2})$. So despite the hardness results of \cref{sec:reduction} we may still hope that even error-prone quantum measurements lead to period finding speedups. Indeed, it is also known that for any fixed $\tau < \frac 1 2 $ the BKW algorithm~\cite{STOC:BluKalWas00} solves $\LPN_{n, \tau}$ --- and thus by \cref{theo:equivalence} also $\LSN_{n, \tau}$ --- in time $2^{\bigO\big(\frac{n}{\log n}\big)}$. This implies that asymptotically the combination of LSN samples together with a suitable LPN-solver already outperforms classical period finding.

In this work, we focus on the LPN-solvers of Esser, Kübler, May~\cite{C:EssKubMay17} rather than the class of BKW-type solvers~\cite{STOC:BluKalWas00,AC:GuoJohLon14,C:KirFou15,C:EHKMS18}, since they have a simple description and runtime analysis, are easy to implement, have low memory consumption, are sufficiently powerful for showing quantum advantage even for large errors $\tau < \frac 12$, and finally they are practically best for the \IBMQ\ error rates $\tau \in [0.09, 0.13]$.

 We start with the analysis of the \textsc{Pooled Gauss} algorithm~\cite{C:EssKubMay17}. \textsc{Pooled Gauss} solves $\LPN_{n, \tau}$ in time  $\tilde{\Theta}\left( 2^{\log\left(\frac{1}{1-\tau}\right)\cdot n} \right)$ using $\tilde{\Theta}\left( n^2 \right)$ samples.

The following theorem shows that period finding with error-prone quantum samples in combination with \textsc{Pooled Gauss} is superior to purely classical period finding whenever the error $\tau$ is bounded by $\tau \leq 0.293$.

\begin{theorem}
\label{theo:pgauss}
In the LSN Error Model (\cref{def:error_model}), \textsc{Pooled Gauss} finds the period $\vec s \in \F_2^n$ of a Simon function $f_{\vec s}$ using $\tilde{\Theta}\left( n^2 \right)$ many $\LSN_{n,\tau}$-samples, coming from practical measurements of Simon's circuit $Q^{\Simon}_{f_{\vec s}}$ with error rate~$\tau$, in time  $\tilde{\Theta}\left( 2^{\log\left(\frac{1}{1-\tau}\right)\cdot n} \right)$. This improves over classical period finding for error rates
\[\tau < 1-\frac{1}{\sqrt 2} \approx 0.293.
\]
\end{theorem}


\begin{proof}
We use \cref{alg:LPNtoLSN}, where any $\bigO_{\LPN}(\vec s)$-call is provided by a measurement of $Q^{\Simon}_{f_{\vec s}}$. In the LSN Error Model, this gives us an $\LSN_{n,\tau}$-instance which is transformed by \cref{alg:LPNtoLSN} into an $\LPN_{n,\tau}$-instance. We use \textsc{Pooled Gauss} as the LPN-solver ${\cal A}_{\LPN}$ inside \cref{alg:LPNtoLSN}. This immediately implies time complexity $\tilde{\Theta}\left( 2^{\log\left(\frac{1}{1-\tau}\right)\cdot n} \right)$.

It remains to show outperformance of the classical algorithm, i.e. $\log\left(\frac{1}{1-\tau}\right)<\frac 1 2$. Notice that our condition $\tau < 1-\frac{1}{\sqrt 2}$ implies that $\frac{1}{1-\tau} < \sqrt{2}$ and therefore
\[
  \log\left(\frac{1}{1-\tau}\right) < \log(\sqrt{2}) = \frac 1 2.
\]
\end{proof}

\cref{theo:pgauss} already shows the usefulness of a quite limited quantum oracle that only allows us polynomially many measurements, whenever its error rate $\tau$ is small enough.

If we allow for more quantum measurements, the \textsc{Well-Pooled Gauss} algorithm~\cite{C:EssKubMay17} solves $\LPN_{n,\tau}$ in improved time and
 query complexity $\tilde{\Theta}( 2^{f(\tau) n})$, where $f(\tau) = 1-\frac{1}{1+\log(\frac{1}{1-\tau})}$. The following theorem shows that \textsc{Well-Pooled Gauss} in combination with error-prone quantum measurements improves on classical period finding for {\em any} error rate $\tau$.

\begin{theorem}
\label{theo:wpgauss}
In the LSN Error Model (\cref{def:error_model}), \textsc{Well Pooled Gauss} finds the period $\vec s \in \F_2^n$ of a Simon function $f_{\vec s}$  using $\tilde{\Theta}( 2^{f(\tau) n})$ many $\LSN_{n,\tau}$-samples, coming from practical measurements of Simon's circuit $Q^{\Simon}_{f_{\vec s}}$ with error rate $\tau$, in time  $\tilde{\Theta}( 2^{f(\tau) n})$, where
\[f(\tau) = 1-\frac{1}{1+\log(\frac{1}{1-\tau})}.\]
This improves over classical period finding for {\em all error rates} $\tau < \frac 1 2$.
\end{theorem}


\begin{proof}
As in the proof of \cref{theo:pgauss} we use \cref{alg:LPNtoLSN}, where measurements of $Q^{\Simon}_{f_{\vec s}}$ provide the $\bigO_{\LPN}(\vec s)$-calls and \textsc{Well Pooled Gauss } is the LPN-solver ${\cal A}_{\LPN}$. Correctness and the claimed complexities follow immediately.

It remains to show outperformance of any classical period finding algorithm. Notice that $\tau < \frac 1 2$ implies $\frac{1}{1-\tau} < 2$ and therefore $\log(\frac{1}{1-\tau}) < 1$. This in turn implies
\[
  f(\tau) = 1-\frac{1}{1+\log(\frac{1}{1-\tau})} < 1-\frac{1}{2} = \frac 1 2.
\]
\end{proof}

The results of \cref{theo:pgauss} and \cref{theo:wpgauss} show that quantum measurements of $Q^{\Simon}_{f_{\vec s}}$ help us (asymptotically) even for large error rates $\tau$, provided that our error model is sufficiently accurate.

%
%

\section{Practical Error Handling for Simon's Algorithm}
\label{sec:cexperiments}
In this section, we compare the practical runtimes needed to find periods $\vec s$ with the smoothed experimental data from our \IBMQ\ quantum measurements (see \cref{fig:simon_smooth_experiment}) with purely classical period finding. 

Notice that our LPN-solvers incur some polynomial overhead, which makes them for very small $n$ as on \IBMQ\ inferior to purely classical period finding. Moreover, we would like to stress the experimental result of \cref{sec:ibmq} that  \IBMQ's error rate $\tau(n)$ is a function increasing in $n$. So even if asymptotically LPN-solvers outperform classical period finding, a fast convergence of $\tau(n)$ towards $\frac 12$ prevents  practical quantum advantage.

\subsubsection{Periods classically.}
Let us start with the description of an optimal classical period finding algorithm, inspired by~\cite{MontanaroW16}. Naively, one may think that it is optimal to query $f_{s}$ at different random points $\vec x_i$, until one hits the first collision $f_{\vec s}(\vec x_i) = f_{\vec s}(\vec x_j)$. However, assume that we have already queried the set of points $P=\{\vec x_1, \vec x_2, \vec x_3\}$,  without obtaining a collision. This gives us the information that $\vec s$ is {\em not} in set of distances $D=\{\vec x_1 + \vec x_2, \vec x_2 + \vec x_3 , \vec x_1 + \vec x_3 \}$. This implies that we should not ask $\vec x_1 + \vec x_2 + \vec x_3$, since it lies at distance $\vec x_1 + \vec x_2$ of $\vec x_3$. Hence on optimal algorithm keeps track of the set $D$ of all excluded distances. This is realized in our algorithm~\textsc{Period}, see \cref{alg:period}.
\begin{algorithm}[htb]
	\DontPrintSemicolon
	\SetAlgoLined
	\SetKwInOut{Input}{Input}\SetKwInOut{Output}{Output}
	
	\SetKwComment{COMMENT}{$\triangleright$\ }{}%
	
	\Input{ Access to $f_{\vec s}$.
	}
	\Output{
		Secret $\vec s$.
	}
	\Begin{
		Set $P=\{(\vec 0,f_{\vec s}(\vec 0))\}$.\COMMENT*{Set of queried points.\hspace*{-1em}}
		Set $D=\{\vec 0\}$. \COMMENT*{Set of distances.\hspace*{-1em}}
		\Repeat{$\exists \vec x'\neq \vec x: (\vec x',f_{\vec s}(\vec x))\in P$ \textbf{\em or }$|D| = 2^n-1$}{
			Select $\vec x\in\text{argmax}\{|\{ \vec s'\in D\mid(\vec x+\vec s', \cdot)\not\in P\}|\}$. \COMMENT*{Optimal next query.\hspace*{-1em}}\label{line:select}
			$P:= P \cup (\vec x,f_{\vec s}(\vec x))$ \COMMENT*{Update queries.\hspace*{-1em}}
			\For{$(\vec x',f_{\vec s}(\vec x')) \in P$}{\label{line:add_loop1}
				$D:= D \cup\{ \vec x+\vec x'\}$ \COMMENT*{Update distances.\hspace*{-1em}} \label{line:add_loop2}
			}\label{line:add_loop3}
		}
		\lIf{$|D| = 2^n-1$}{
			\KwRet{$\vec s\in \F_2^n\setminus D$.} \hfill\COMMENT*[h]{Only possible period.\hspace*{-1em}}\hspace*{-0.34em}
		}
		\lElse{
			\KwRet{$\vec x+\vec x'$.} \hfill\COMMENT*[h]{Collision found.\hspace*{-1em}}\hspace*{-0.34em}
		}
	}
	\caption{\textsc{Period}\label{alg:period}}
	
\end{algorithm}
\subsubsection{Periods quantumly.} By the result of \cref{sec:reduction} we may first transform our quantum measurements into \LPN\ samples, and then use one of the \LPN-solvers from \cref{sec:correct}.
Since the error rates from our \textit{smoothed} IBM-Q16 measurements (\cref{fig:simon_smooth_experiment}) are below $\frac 1 8$, according to \cref{theo:pgauss}  we may  use \textsc{Pooled Gauss}.

Instead of applying the \LSN-to-\LPN\ reduction to our smoothed data, we directly adapt \textsc{Pooled Gauss} into an \LSN-solver, called \textsc{Pooled LSN} (\cref{alg:PS}).
\textsc{Pooled LSN} can be considered as a fault-tolerant version of \textsc{Simon} (\cref{alg:simon}) that iterates until we obtain an error-free set of $n-1$ linearly independent vectors. Notice that error-freeness can be tested, since the resulting potential period $\vec s'$ is correct iff $f_{\vec s}(\vec s'){=}f_{\vec s}(\vec 0)$.

%
\begin{algorithm}[htb]
	\DontPrintSemicolon
	\SetAlgoLined
	\SetKwInOut{Input}{Input}\SetKwInOut{Output}{Output}

	\SetKwComment{COMMENT}{$\triangleright$\ }{}

	\Input{ Pool $P \subset \F_2^n$ of \LSN\ samples with $|P| \geq n-1$
	}
	\Output{
		Secret $\vec s$
	}
	\Begin{
		\Repeat{$f_{\vec s}(\vec s')\overset{?}{=}f_{\vec s}(\vec 0)$}{
			Randomly select a linearly independent set $Y = \{\vec y_1, \ldots ,  \vec y_{n-1} \} \subseteq P$.\;
			Compute the unique $\vec s' \in Y^{\perp} \setminus \{\vec 0\}$.\;
		}
		\KwRet{$\vec s'$.}
	}
	\caption{\textsc{Pooled LSN}\label{alg:PS}}
\end{algorithm}

\vspace*{-1em}
\subsubsection{Run time comparison.} \textsc{Period} and \textsc{Pooled LSN} exponentially often iterate their \textbf{repeat}-loops, where each iteration runs in polynomial time (using the right data structure). Hence, asymptotically the number of iterations dominate runtimes for both algorithms. For ease of simplicity, we take as cost measure only the exponential number of loops, ignoring all polynomial factors (the polynomial factors actually dominate in practice for our small dimensions $n$). 
\begin{wrapfigure}{r}{0.5\textwidth}
		\vspace*{-1em}
	\centering
	\hspace*{-0.25cm}
	\begin{tikzpicture}[scale=0.75]
		\begin{axis}[
			legend style={legend pos=north west ,cells={align=left}},
			xlabel=Dimension,
			ylabel=Iteration (log),
			]
			
			\addplot+[] coordinates {
				(2,1.14158)
				(3,1.70256)
				(4,2.05105)
				(5,2.37014)
				(6,2.64965)
				(7,2.91846)
			};
			\addlegendentry{\textsc{Pooled LSN} loops}
			\addplot+[] coordinates {
				(2,0.7350021934953096)
				(3,1.274708939214712)
				(4,1.8742854904980353)
				(5,2.4584869999743257)
				(6,2.9637515386552695)
				(7,3.4615408330178317)
			};
			\addlegendentry{\textsc{Period} loops}
		\end{axis}
	\end{tikzpicture}%
	\caption{Log-scaled loop iterations of \textsc{Pooled LSN} and \textsc{Period}, averaged over 10.000 iterations. }\vspace*{-1em}
	\label{fig:plot:runtime}
	\vspace*{-1em}
\end{wrapfigure}
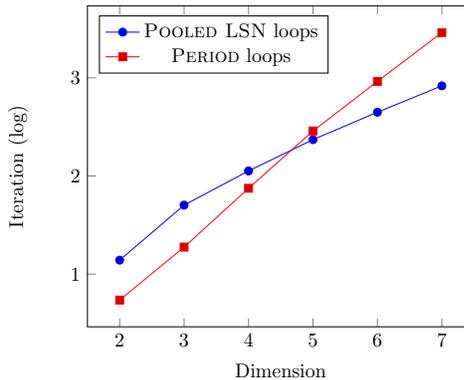
Using this (over-)simplified loop cost measure, we ran 10.000 iterations of \textsc{Period} for $n=2, \ldots, 7$ and averaged over the runtimes. For the quantum period finding we took as pool $P$ the complete smoothed data of \cref{fig:simon_smooth_experiment}.
We then also ran 10.000 iterations of \textsc{Pooled LSN} for $n=2, \ldots, 7$ and averaged over the runtimes.
The resulting log-scaled runtimes are depicted in  \cref{fig:plot:runtime}.

As expected, \textsc{Period}'s experimental runtime exponent is $\frac n 2$. For \textsc{Pooled LSN}, we obtain an experimental regression line of roughly $\frac n 3$, where the slope seems to decrease with $n$. This results in a cut-off point for the loop numbers between $n=4$ and $n=5$. Thus, experimentally we obtain quantum advantage, at least for our loop cost measure. 

\paragraph{Acknowledgement.}
We acknowledge use of the IBM Q for this work. The views expressed are those of the authors and do not reflect the official policy or position of IBM or the IBM Q team.

%
%

\bibliographystyle{Bib/splncs03}
\bibliography{Bib/abbrev3,Bib/crypto,Bib/IBM}

\begin{thebibliography}{10}
\providecommand{\url}[1]{\texttt{#1}}
\providecommand{\urlprefix}{URL }

\bibitem{IBMQ16}
15-qubit backend: IBM Q team, "IBM Q 16 Melbourne backend specification
  V2.0.1," (2020). Retrieved from https://quantum-computing.ibm.com. Accessed
  14. January 2020.

\bibitem{EC:AlaRus17}
Alagic, G., Russell, A.: Quantum-secure symmetric-key cryptography based on
  hidden shifts. In: Coron, J., Nielsen, J.B. (eds.) EUROCRYPT~2017, Part~III.
  {LNCS}, vol. 10212, pp. 65--93. Springer, Heidelberg (Apr~/~May 2017)

\bibitem{FOCS:Alekhnovich03}
Alekhnovich, M.: More on average case vs approximation complexity. In: 44th
  FOCS. pp. 298--307. {IEEE} Computer Society Press (Oct 2003)

\bibitem{STOC:BluKalWas00}
Blum, A., Kalai, A., Wasserman, H.: Noise-tolerant learning, the parity
  problem, and the statistical query model. In: 32nd ACM STOC. pp. 435--440.
  {ACM} Press (May 2000)

\bibitem{C:BonZha13}
Boneh, D., Zhandry, M.: Secure signatures and chosen ciphertext security in a
  quantum computing world. In: Canetti, R., Garay, J.A. (eds.) CRYPTO~2013,
  Part~II. {LNCS}, vol. 8043, pp. 361--379. Springer, Heidelberg (Aug 2013)

\bibitem{SAC:Bonnetain17}
Bonnetain, X.: Quantum key-recovery on full {AEZ}. In: Adams, C., Camenisch, J.
  (eds.) SAC 2017. {LNCS}, vol. 10719, pp. 394--406. Springer, Heidelberg (Aug
  2017)

\bibitem{DBLP:asiacrypt19}
Bonnetain, X., Hosoyamada, A., Naya-Plasencia, M., Sasaki, Y., Schrottenloher,
  A.: Quantum attacks without superposition queries: the offline simon's
  algorithm. In: Advances in Cryptology - {ASIACRYPT} 2019 (2019)

\bibitem{AC:BonNay18}
Bonnetain, X., {Naya-Plasencia}, M.: Hidden shift quantum cryptanalysis and
  implications. In: Peyrin, T., Galbraith, S. (eds.) ASIACRYPT~2018, Part~I.
  {LNCS}, vol. 11272, pp. 560--592. Springer, Heidelberg (Dec 2018)

\bibitem{calderbank1997quantum}
Calderbank, A.R., Rains, E.M., Shor, P.W., Sloane, N.J.: Quantum error
  correction and orthogonal geometry. Physical Review Letters  78(3),  405
  (1997)

\bibitem{C:EHKMS18}
Esser, A., Heuer, F., K{\"u}bler, R., May, A., Sohler, C.: Dissection-{BKW}.
  In: Shacham, H., Boldyreva, A. (eds.) CRYPTO~2018, Part~II. {LNCS}, vol.
  10992, pp. 638--666. Springer, Heidelberg (Aug 2018)

\bibitem{C:EssKubMay17}
Esser, A., K{\"u}bler, R., May, A.: {LPN} decoded. In: Katz, J., Shacham, H.
  (eds.) CRYPTO~2017, Part~II. {LNCS}, vol. 10402, pp. 486--514. Springer,
  Heidelberg (Aug 2017)

\bibitem{AC:EveMan91}
Even, S., Mansour, Y.: A construction of a cipher from a single pseudorandom
  permutation. In: Imai, H., Rivest, R.L., Matsumoto, T. (eds.) ASIACRYPT'91.
  {LNCS}, vol. 739, pp. 210--224. Springer, Heidelberg (Nov 1993)

\bibitem{AC:GuoJohLon14}
Guo, Q., Johansson, T., L{\"o}ndahl, C.: Solving {LPN} using covering codes.
  In: Sarkar, P., Iwata, T. (eds.) ASIACRYPT~2014, Part~I. {LNCS}, vol. 8873,
  pp. 1--20. Springer, Heidelberg (Dec 2014)

\bibitem{RSA:HosSas18}
Hosoyamada, A., Sasaki, Y.: Cryptanalysis against symmetric-key schemes with
  online classical queries and offline quantum computations. In: Smart, N.P.
  (ed.) CT-RSA~2018. {LNCS}, vol. 10808, pp. 198--218. Springer, Heidelberg
  (Apr 2018)

\bibitem{C:KLLN16}
Kaplan, M., Leurent, G., Leverrier, A., {Naya-Plasencia}, M.: Breaking
  symmetric cryptosystems using quantum period finding. In: Robshaw, M., Katz,
  J. (eds.) CRYPTO~2016, Part~II. {LNCS}, vol. 9815, pp. 207--237. Springer,
  Heidelberg (Aug 2016)

\bibitem{C:KirFou15}
Kirchner, P., Fouque, P.A.: An improved {BKW} algorithm for {LWE} with
  applications to cryptography and lattices. In: Gennaro, R., Robshaw, M.J.B.
  (eds.) CRYPTO~2015, Part~I. {LNCS}, vol. 9215, pp. 43--62. Springer,
  Heidelberg (Aug 2015)

\bibitem{KuwakadoM12}
Kuwakado, H., Morii, M.: Security on the quantum-type even-mansour cipher. In:
  Proceedings of the International Symposium on Information Theory and its
  Applications, {ISITA} 2012, Honolulu, HI, USA, October 28-31, 2012. pp.
  312--316 (2012), \url{http://ieeexplore.ieee.org/document/6400943/}

\bibitem{AC:LeaMay17}
Leander, G., May, A.: Grover meets simon - quantumly attacking the
  {FX}-construction. In: Takagi, T., Peyrin, T. (eds.) ASIACRYPT~2017, Part~II.
  {LNCS}, vol. 10625, pp. 161--178. Springer, Heidelberg (Dec 2017)

\bibitem{MontanaroW16}
Montanaro, A., de~Wolf, R.: A survey of quantum property testing. Theory of
  Computing, Graduate Surveys  7,  1--81 (2016),
  \url{https://doi.org/10.4086/toc.gs.2016.007}

\bibitem{STOC:Regev05}
Regev, O.: On lattices, learning with errors, random linear codes, and
  cryptography. In: Gabow, H.N., Fagin, R. (eds.) 37th ACM STOC. pp. 84--93.
  {ACM} Press (May 2005)

\bibitem{SantoliSchaffner16}
Santoli, T., Schaffner, C.: Using simon's algorithm to attack symmetric-key
  cryptographic primitives. Quantum Information {\&} Computation  17(1{\&}2),
  65--78 (2017),
  \url{http://www.rintonpress.com/xxqic17/qic-17-12/0065-0078.pdf}

\bibitem{ShendePMH03}
Shende, V.V., Prasad, A.K., Markov, I.L., Hayes, J.P.: Synthesis of reversible
  logic circuits. {IEEE} Trans. on {CAD} of Integrated Circuits and Systems
  22(6),  710--722 (2003), \url{https://doi.org/10.1109/TCAD.2003.811448}

\bibitem{FOCS:Shor94}
Shor, P.W.: Algorithms for quantum computation: Discrete logarithms and
  factoring. In: 35th FOCS. pp. 124--134. {IEEE} Computer Society Press (Nov
  1994)

\bibitem{FOCS:Simon94}
Simon, D.R.: On the power of quantum computation. In: 35th FOCS. pp. 116--123.
  {IEEE} Computer Society Press (Nov 1994)

\bibitem{PhysRevLett.113.200501}
Tame, M.S., Bell, B.A., Di~Franco, C., Wadsworth, W.J., Rarity, J.G.:
  Experimental realization of a one-way quantum computer algorithm solving
  simon's problem. Phys. Rev. Lett.  113,  200501 (Nov 2014),
  \url{https://link.aps.org/doi/10.1103/PhysRevLett.113.200501}

\bibitem{CHES:YZSAG15}
Yang, G., Zhu, B., Suder, V., Aagaard, M.D., Gong, G.: The simeck family of
  lightweight block ciphers. In: G{\"u}neysu, T., Handschuh, H. (eds.)
  CHES~2015. {LNCS}, vol. 9293, pp. 307--329. Springer, Heidelberg (Sep 2015)

\end{thebibliography}
\newpage
\appendix
\section{Appendix}
\label{sec:appendix}
\begin{figure}[htb]
	\begin{center}
		\begin{subfigure}[t]{0.49\textwidth}
			\centering
			\includegraphics[width=0.65\textwidth]{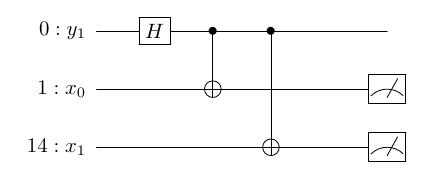}
			\caption{$n = 2$}
		\end{subfigure}
		\hfill
		\begin{subfigure}[t]{0.49\textwidth}
			\centering
			\includegraphics[width=0.65\textwidth]{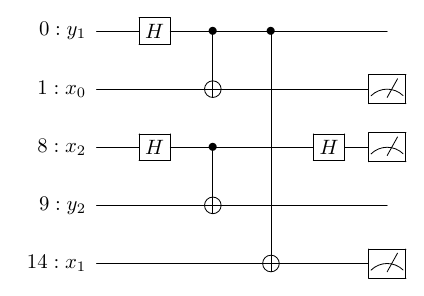}
			\caption{$n = 3$}
		\end{subfigure}
		\\
		\begin{subfigure}[t]{0.49\textwidth}
			\centering
			\includegraphics[width=0.65\textwidth]{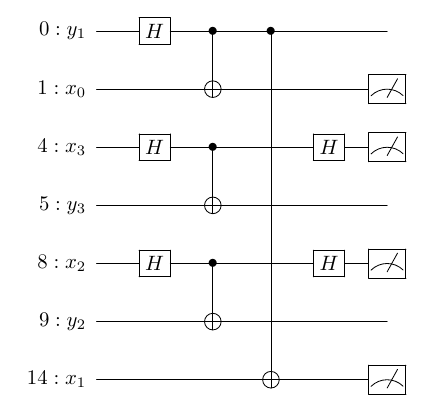}
			\caption{$n = 4$}
		\end{subfigure}
		\hfill
		\begin{subfigure}[t]{0.49\textwidth}
			\centering
			\includegraphics[width=0.65\textwidth]{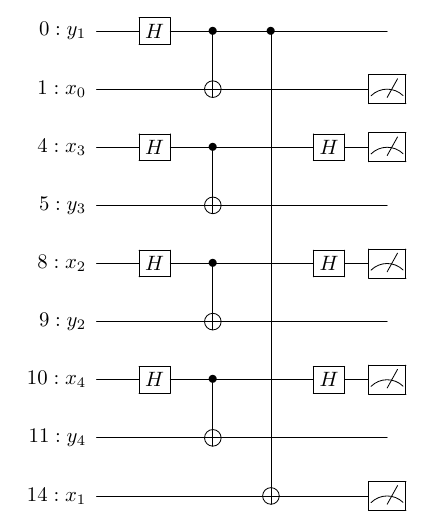}
			\caption{$n = 5$}
		\end{subfigure}
		\\
		\begin{subfigure}[t]{0.49\textwidth}
			\centering
			\includegraphics[width=0.65\textwidth]{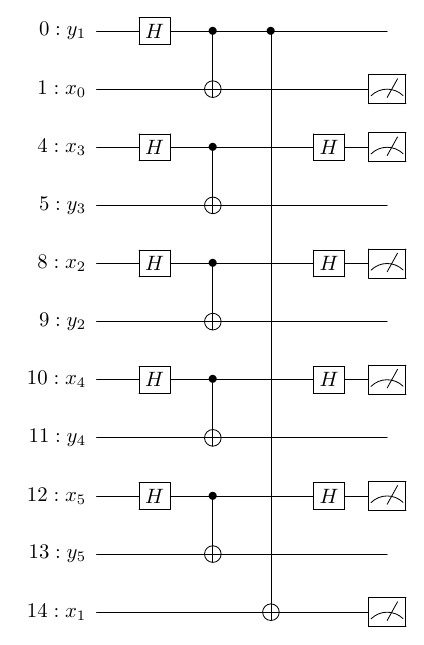}
			\caption{$n = 6$}
		\end{subfigure}
		\hfill
		\begin{subfigure}[t]{0.49\textwidth}
			\centering
			\includegraphics[width=0.65\textwidth]{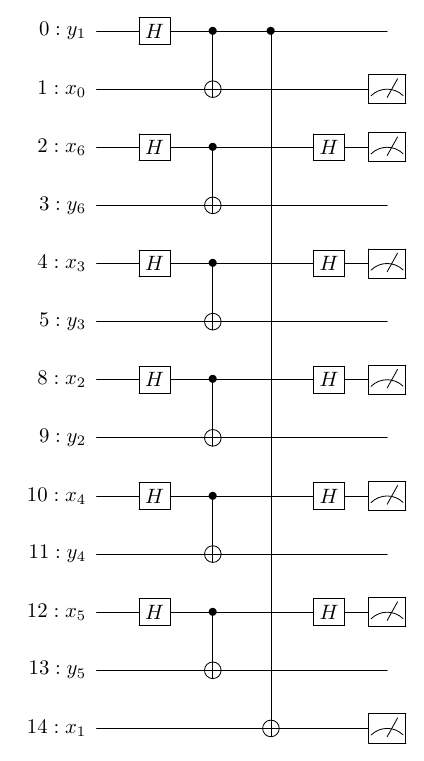}
			\caption{$n = 7$}
		\end{subfigure}
	\end{center}
	\vspace*{-0.5cm}
	\caption{Optimized circuits for $n=2, \ldots, 7$ with $\vec s = 0^{n-2}11$. We omit qubits $6, 7$ and so input $y_0$, which is not required after optimization.}\label{fig:collection_circuit}
	\vspace*{-3cm}
\end{figure}
\end{document}